\documentclass[12pt,reqno]{amsart}

\usepackage[left=1in, right=1in, top=1.1in,bottom=1.1in]{geometry}
\usepackage[title]{appendix}

\usepackage{graphicx}
\graphicspath{ {./fig/}}
\usepackage[nomarkers, nolists, figuresonly]{endfloat}

\usepackage[T1]{fontenc}
\usepackage{amsthm,amssymb,amsmath,amsfonts,mathrsfs,amscd}
\usepackage{enumitem}
\usepackage{stmaryrd}
\usepackage{mathrsfs}

\usepackage{pdfsync}
\usepackage[font={scriptsize}]{caption}

\usepackage{hyperref}
\hypersetup{
  colorlinks   = true,
  citecolor    = green,
  linkcolor    = red
}

\newtheorem{theorem}{Theorem}[section]

\newtheorem{corollary}[theorem]{Corollary}

\newtheorem{definition}[theorem]{Definition}

\newtheorem{lemma}[theorem]{Lemma}

\newtheorem{proposition}[theorem]{Proposition}

\theoremstyle{remark}
\newtheorem{remark}[theorem]{Remark}

\theoremstyle{remark}

\newtheorem{assumption}[theorem]{Assumption}

\newcommand{\Eof}[1]{\mathbb{E}\left[ #1 \right]}
\newcommand{\p}{\partial}
\newcommand{\cL}{\mathcal{L}}

\title{Growth rate of liquidity provider's wealth in G3Ms}
\author[C.~Y. Lee]{Cheuk Yin Lee}
\author[S.-N. Tung]{Shen-Ning Tung}
\author[T.-H. Wang]{Tai-Ho Wang}
\date{\today}

\address{Cheuk Yin Lee \newline
School of Science and Engineering, \newline
The Chinese University of Hong Kong (Shenzhen)\newline
Shenzhen, China
}
\email{leecheukyin@cuhk.edu.cn}

\address{Shen-Ning Tung \newline
Department of Mathematics, \newline
National Tsing Hua University \newline
Hsinchu, Taiwan
}
\email{tung@math.nthu.edu.tw}

\address{Tai-Ho Wang \newline
Department of Mathematics \newline
Baruch College, The City University of New York \newline
1 Bernard Baruch Way, New York, NY10010
}
\email{tai-ho.wang@baruch.cuny.edu}


\begin{document}

\begin{abstract}
We study how trading fees and continuous-time arbitrage affect the profitability of liquidity providers (LPs) in Geometric Mean Market Makers (G3Ms). We use stochastic reflected diffusion processes to analyze the dynamics of a G3M model under the arbitrage-driven market \cite{milionis2022automated}. Our research focuses on calculating LP wealth and extends the findings of Tassy and White \cite{tassy2020growth} for the constant product market maker (Uniswap v2) to a broader range of G3Ms, including Balancer. This allows us to calculate the long-term expected logarithmic growth of LP wealth, offering new insights into the complex dynamics of AMMs and their implications for LPs in decentralized finance.
\end{abstract}

\maketitle

\section{Introduction}

% ---------------------------------------------------------------
% 1. OPENING: Introduction to AMMs and G3Ms
% ---------------------------------------------------------------
Decentralized finance (DeFi) has transformed financial markets by enabling trustless trading and investment through blockchain technology
\cite{capponi2021adoption, gobet2023decentralized}. At the heart of this
transformations are Automated Market Makers (AMMs) \cite{mohan2022automated}, which replace traditional order book exchanges with smart contracts that automate price discovery and trade execution, eliminating intermediaries and broadening market accessibility. Among the diverse landscape of AMM designs, Geometric Mean Market Makers (G3Ms) have gained significant traction, powering leading DeFi protocols such as Uniswap \cite{adams2020uniswap} and Balancer \cite{martinelli2019balancer}. By maintaining a fixed weighted geometric mean of asset reserves, a G3M enforces a constant target allocation across its constituent assets, creating a natural rebalancing dynamic that closely mirrors the discipline of a classical constant-weighted index fund. This structural parallel raises a natural and practically important question: can a G3M, through its fee mechanism, generate long-term wealth growth that matches or even exceeds that of a traditional passive investment strategy?

% ---------------------------------------------------------------
% 2. MOTIVATION AND METHODOLOGY
% ---------------------------------------------------------------
The goal of this paper is to address this question by developing a tractable analytical framework for studying the long-term growth of liquidity provider (LP) wealth in G3Ms, explicitly accounting for the interplay between continuous-time arbitrage and transaction fees. Our approach is motivated by the observation that, in a G3M, every trade moves the pool price toward the external market price, generating fee revenue for LPs while simultaneously exposing them to adverse selection. Understanding how these two forces interact over long horizons is essential for evaluating G3Ms as investment vehicles and for designing fee structures that maximize LP returns.

We model the G3M under the influence of continuous arbitrage activity against a frictionless external reference market, adapting the arbitrage-driven price dynamics of \cite{najnudel2024arbitrage} and drawing on the loss-versus-rebalancing framework of \cite{milionis2022quantifying, milionis2022automated, milionis2023automated}. The key analytical tool is the theory of reflected diffusion processes \cite{lions1984stochastic}: the log-ratio of the external market price to the G3M pool price, which we call the \emph{mispricing process}, evolves as a diffusion confined to the no-arbitrage interval $[\ln\gamma, -\ln\gamma]$ by reflecting boundary local times that encode the fee-driven arbitrage activity. The long-term growth rate of LP wealth can then be expressed in terms of the ergodic properties of this reflected diffusion, yielding explicit formulas that connect fee structures, market volatility, and portfolio weights.

The scope of this paper is not to provide a fully realistic microstructural model of AMM trading, but to establish a clean baseline that isolates the arbitrage channel and demonstrates the potential of G3Ms to serve as on-chain index fund infrastructure. In particular, our results show that, from a portfolio construction perspective, short-term arbitrage losses can be transformed into long-run volatility harvesting through the fee mechanism, with the LP wealth growth rate outperforming the excess growth rate of a constant-rebalanced portfolio in Stochastic Portfolio Theory (SPT) \cite{fernholz2002stochastic, karatzas2009stochastic} --- which itself assumes frictionless trading --- when the fee tier is chosen appropriately.

To this end, we adopt three simplifying assumptions, each of which we now motivate.

\paragraph{\textit{Assumption A: No Noise Traders.}}
We assume all trades are executed by arbitrageurs, with no noise traders present. This is a deliberate worst-case baseline: in practice, retail order flow is typically routed through aggregation solvers such as 1inch or CowSwap, whose highly nonlinear routing logic is difficult to incorporate into a closed-form model (see, e.g., \cite{brady2024amm}). Under this assumption, LPs face only adversarial order flow and are in the least favorable scenario for fee collection. Nevertheless, our results demonstrate that even in this worst-case setting, the LP wealth growth rate can outpace that of a constant-rebalanced portfolio under SPT, provided the fee tier is
chosen appropriately. The presence of uninformed order flow would, in practice, only strengthen this conclusion.

\paragraph{\textit{Assumption B: Frictionless External Reference Market.}}
We assume the existence of an external reference price with infinite liquidity and no transaction costs, representing an aggregated oracle across all available liquidity sources --- DEXs, CEXs, and routing solvers --- rather than the price on any single venue. The special case of a single CEX with proportional transaction costs is studied in \cite{campbell2025optimal}, where the arbitrage structure takes a similar form to ours but with a modified no-arbitrage boundary; in that setting, our growth rate computation carries over verbatim. Incorporating fixed transaction costs, however, would break the scale homogeneity of the mispricing dynamics, substantially increasing analytical complexity. We therefore leave a systematic treatment of inter-venue transaction costs
as an important direction for future work.

\paragraph{\textit{Assumption C: Continuous Arbitrage.}}
We assume arbitrageurs continuously monitor and act on arbitrage opportunities instantaneously. While idealized, this assumption is increasingly justified by modern blockchain infrastructure: block times on networks such as Arbitrum ($\approx 250$ms), Solana ($\approx 400$ms), and MegaETH ($\approx 10$ms) are short enough that the continuous-time limit provides a reasonable approximation of high-frequency arbitrage activity.

% ---------------------------------------------------------------
% 3. LITERATURE REVIEW
% ---------------------------------------------------------------
\subsection{Literature Review}

\paragraph{\textit{AMMs.}}
A large body of literature has examined various aspects of AMM mechanics.
\cite{angeris2019analysis, angeris2020improved} analyzed the theoretical properties and price-tracking capabilities of constant function markets.
\cite{evans2021liquidity} extended this analysis to G3Ms with variable weights, while \cite{fukasawa2023weighted, fukasawa2023modelfree} investigated impermanent loss and hedging strategies. More recently, \cite{najnudel2024arbitrage} examined arbitrage-driven price dynamics in fee-based AMMs, providing insights into the relationship between fee structures and market behavior.

\paragraph{\textit{LP Risk.}}
A central concern in the AMM literature is the quantification and management of LP losses. \cite{milionis2022quantifying, milionis2022automated, milionis2023automated} introduced the concept of ``loss-versus-rebalancing'' (LVR) and developed frameworks for measuring systematic LP losses in both standard AMMs and concentrated liquidity markets. \cite{cartea2023predictable, cartea2024decentralized} explored predictable
loss in constant function markets and developed optimal liquidity provision strategies under separate fee rate models. \cite{bronnimann2024automated} further analyzed risks and incentives in AMM liquidity provision, proposing novel transaction cost structures. \cite{risk2026pricing} develops a pricing and hedging framework for CFMMs, characterizing impermanent loss as a weighted strip of vanilla options via the Carr--Madan spanning formula and validating the resulting implied volatility structure against Uniswap v3
and Deribit option data. Beyond loss quantification, \cite{drissi2025equilibrium} studied equilibrium liquidity provision in decentralized markets, demonstrating how LPs can offset inventory and adverse selection risks through hedging strategies, and highlighting the importance of fee income in the context of LP risk-return trade-offs.

\paragraph{\textit{Blockchain Market Microstructure.}}
A growing literature examines microstructural features specific to blockchain trading environments that are absent from classical market models. Factors such as gas costs, block discreteness, and market fragmentation can materially affect LP returns and arbitrage dynamics. \cite{aoyagi2025coexisting} examines the interaction between AMMs
and centralized limit order books, demonstrating how their coexistence influences liquidity provision and price discovery. \cite{capponi2025viability} develops a model to evaluate the viability of blockchain markets as the sole venue for price formation, with particular attention to the effects of block times on market efficiency. Closely
related to our analysis of optimal fee tiers, \cite{campbell2025optimal} studies LP profitability in a dynamic model where an AMM operates alongside a CEX with optimal order routing and arbitrage, finding that the optimal fee is stable under normal market conditions but should be raised during periods of high volatility --- a finding that complements our characterization of the fee tier that maximizes long-term LP wealth growth.

\subsection{Contribution}
This paper extends the analysis of \cite{tassy2020growth}, which studied LP wealth growth in the constant product market maker (Uniswap V2), to a broad class of G3Ms, including Balancer. Within the framework described above, the paper makes three key contributions:
\begin{enumerate}
    \item We derive explicit formulas for the long-term expected logarithmic growth rate of LP wealth under various market conditions, including time-homogeneous, time-inhomogeneous, and stochastic volatility settings, and across different volatility and drift scenarios.
    \item We analyze how different fee structures affect LP returns, characterizing the optimal fee tier that maximizes LP wealth growth and revealing non-monotone phase transitions as a function of pool weights and market drift.
    \item We compare the long-term performance of G3Ms to constant-rebalanced
    portfolio strategies in the sense of SPT, demonstrating that G3Ms can serve as a competitive on-chain index fund infrastructure capable of outperforming classical passive investment benchmarks.
\end{enumerate}

\subsection*{Outline}
The remainder of the paper is organized as follows. Section \ref{sec:G3Ms}
establishes notation and examines G3M market dynamics with and without transaction costs. Section \ref{sec:arbitrage_driven_dynamics} studies the arbitrageur's problem and G3M dynamics under both discrete and continuous arbitrage. Section \ref{sec:LP_growth_rate} presents our main results, connecting LP wealth growth to parabolic PDEs with Neumann boundary conditions and extending the analysis to cases with independent stochastic volatility and drift.
\section{Constant Weight G3Ms} \label{sec:G3Ms}
Geometric Mean Market Makers (G3Ms) are a prominent class of automated market makers (AMMs) that utilize a weighted geometric mean to define the feasible set of trades. This section provides a detailed review of G3M mechanisms, following the work in \cite{evans2021liquidity, mohan2022automated, fukasawa2023modelfree}, for a system with two assets, $X$ and $Y$, and a fixed weight $w \in (0,1)$.

\subsection{G3M Trading Mechanics}
Let $(x,y) \in \mathbb{R}^{2}_+$ denote the reserves of assets $X$ and $Y$ in the liquidity pool (LP). The core principle of a G3M is to maintain a constant weighted geometric mean of the reserves:
\begin{equation} \label{eqn:geometric_mean}
    x^w y^{1-w} = \ell,
\end{equation}
where $\ell$ represents the overall liquidity in the pool.

\subsubsection{Trading without Transaction Costs}
In an idealized setting without transaction costs, trades in a G3M must preserve the constant weighted geometric mean of the reserves. This means a trade $\Delta = (\Delta_x, \Delta_y) \in \mathbb{R}^2$, representing the changes in asset reserves, is feasible if and only if:
\begin{equation} \label{eqn:feasible_trade}
x^w y^{1-w} = (x + \Delta_x)^w (y + \Delta_y)^{1-w}.
\end{equation}
Here, $\Delta_x > 0$ indicates that asset $X$ is being added to the pool, while $\Delta_x < 0$ means $X$ is being withdrawn.

To determine the price of asset $X$ relative to asset $Y$, we analyze how infinitesimal changes in reserves affect the weighted geometric mean. Taking the total derivative of equation \eqref{eqn:geometric_mean}, we get
$$
w x^{w-1} y^{1-w} dx + (1-w) x^w y^{-w} dy = 0,
$$
which simplifies to
$$
w \frac{dx}{x} + (1-w) \frac{dy}{y} = 0.
$$
This relationship between infinitesimal changes in $x$ and $y$ allows us to define the price $P$ of asset $X$ in terms of asset $Y$:
\begin{equation} \label{eqn:AMM_price}
    P = \left. -\frac{d \Delta_y}{d \Delta_x} \right|_{\Delta_x=0} = - \frac{dy}{dx} = \frac{y/(1-w)}{x/w} = \frac{w}{1-w} \frac{y}{x}.
\end{equation}
Therefore, the price in a G3M without transaction costs is determined solely by the ratio of the reserves.

\begin{remark}
A key advantage of G3Ms without transaction costs is their \textit{path independence} \cite[\S 2.3]{angeris2020improved}. This means the final state of the pool depends only on the net change in reserves, not the specific sequence of trades that led to it. This property makes these G3Ms resistant to price manipulation strategies.
\end{remark}

\subsubsection{With Transaction Costs} \label{sec:w_transaction_cost}
In a more realistic scenario, there are transaction costs, typically modelled as a proportional fee. Let $\gamma \in (0,1)$ represent the transaction cost parameter, where $1-\gamma$ is the proportional fee. The trading process now involves two key concepts: feasible trades and reserve updates.

\paragraph{\textit{Determining Feasible Trades}}
A trade is feasible if it satisfies the following conditions, which incorporate the transaction cost:
\begin{itemize}
    \item \textbf{Buying Asset $X$ from the Pool ($\Delta_y > 0$):} The trader pays a fee on the amount of asset $Y$ they provide.
    \begin{equation} \label{eqn:AMM_buy}
        (x + \Delta_x)^w (y + \gamma \Delta_y)^{1-w} = \ell.
    \end{equation}
    \item \textbf{Selling Asset $X$ to the Pool ($\Delta_x > 0$):} The trader pays a fee on the amount of asset $X$ they provide.
    \begin{equation} \label{eqn:AMM_sell}
        (x + \gamma \Delta_x)^w (y + \Delta_y)^{1-w} = \ell.
    \end{equation}
\end{itemize}

\paragraph{\textit{Updating Reserves}}
After a trade, the reserves and the pool's liquidity are updated accordingly:
\begin{align*}
    x &\mapsto x + \Delta_x, \\
    y &\mapsto y + \Delta_y, \\
    \ell &\mapsto (x + \Delta_x)^w (y + \Delta_y)^{1-w}.
\end{align*}

Introducing transaction costs leads to several important differences:

\paragraph{\textit{Marginal Exchange Rate}}
The effective price for infinitesimal trades is given by the marginal exchange rate, which now incorporates the transaction cost.
\begin{itemize}
\item For buying $X$: Differentiating \eqref{eqn:AMM_buy} gives $-\frac{d \Delta_y}{d \Delta_x}|_{\Delta_x=0} = \frac{1}{\gamma}P$.
\item For selling $X$: Differentiating \eqref{eqn:AMM_sell} gives $-\frac{d \Delta_y}{d \Delta_x}|_{\Delta_x=0} = \gamma P$.
\end{itemize}
This creates a \textbf{bid-ask spread}, where the price for buying $X$ is higher than the price for selling $X$.

\paragraph{\textit{Constant Wealth Proportion}}
Despite the transaction costs, the proportion of asset $X$ (and $Y$) in the pool's total wealth remains constant at $w$ (and $1-w$, respectively), as shown by \eqref{eqn:AMM_price}:
\begin{equation} \label{eqn:weight}
    \frac{Px}{w} = \frac{y}{1-w}.
\end{equation}

\paragraph{\textit{Relationship between Reserves and Liquidity}}
There is a correspondence between $(x,y)$ and $(\ell,P)$ given by:
\begin{align*}
    \ln x &= \ln \ell - (1-w) \ln P - (1-w) \ln(1-w) + (1-w) \ln w, \\
    \ln y &= \ln \ell + w \ln P + w \ln(1-w) - w \ln w.
\end{align*}

\paragraph{\textit{Price Impact of Trading}}
The impact of trading on the pool price is given by:
\begin{align*}
    \frac{dP}{P} = \frac{dy}{y} - \frac{dx}{x} = - \frac{1}{1-w} \frac{dx}{x} = \frac{1}{w} \frac{dy}{y}.
\end{align*}

\begin{remark} \
    \begin{enumerate}
        \item The transaction cost parameter $\gamma$ creates a bid-ask spread, similar to traditional limit order books, where buyers pay a slightly higher price than sellers.
        \item Unlike G3Ms without transaction fees, the presence of transaction costs introduces \textbf{path dependence} \cite[\S 2.3]{angeris2020improved}. This means that the order and size of trades influence the final outcome. To illustrate this, consider a trader selling $\Delta_x$ of asset $X$ in exchange for $\Delta_y$ of asset $Y$. The transaction satisfies
        $$
        (x+ \gamma\Delta_x)^w (y+\Delta_y)^{1-w} = x^w y^{1-w}.
        $$
        If the trader instead splits the trade into two smaller transactions, $\Delta_x = \Delta^1_x + \Delta^2_x$, receiving $\Delta^1_y$ and $\Delta^2_y$ of asset $Y$ respectively, then the trades satisfy
        \begin{align*}
        (x+ \gamma\Delta^1_x)^w (y+\Delta^1_y)^{1-w} &= x^w y^{1-w}, \\
        (x+ \Delta^1_x + \gamma \Delta^2_x)^w (y+\Delta^1_y+\Delta^2_y)^{1-w} &= (x+ \Delta^1_x)^w (y+\Delta^1_y)^{1-w}.
        \end{align*}
        Comparing these equations reveals that $\Delta_y < \Delta^1_y + \Delta^2_y <0$, demonstrating that splitting the trade leads to a higher cost (i.e., a smaller amount of asset Y received). This phenomenon is further explored in \cite[Appendix D]{angeris2019analysis} and \cite[Proposition 1]{fukasawa2023modelfree}.
    \end{enumerate}
\end{remark}

\subsection{Continuous Trading Dynamics}
In the continuous trading regime, where trades occur infinitesimally often, we can describe the G3M dynamics using differential equations. These equations capture how the reserves, price, and liquidity evolve in response to continuous buying and selling.

\subsubsection{Reserve Dynamics}
The changes in reserves $x$ and $y$ are governed by the following equations, which incorporate the transaction cost parameter $\gamma$:
\begin{itemize}
    \item \textbf{When buying $X$ from the pool ($dx<0$):}
    \begin{equation} \label{eqn:xy_buy}
        w \frac{dx}{x} + \gamma(1-w)\frac{dy}{y} = 0.
    \end{equation}
    \item \textbf{When selling $X$ to the pool ($dx>0$):}
    \begin{equation} \label{eqn:xy_sell}
        \gamma w \frac{dx}{x} + (1-w)\frac{dy}{y} = 0.
    \end{equation}
\end{itemize}
These equations reflect that when buying $X$, the trader pays a fee on the $Y$ asset provided, while when selling $X$, the fee is paid on the $X$ asset provided.

\subsubsection{Price Dynamics}
The price $P$ of asset $X$ relative to asset $Y$ evolves according to:
\begin{equation} \label{eqn:price_diff}
    \frac{dP}{P} = \frac{dy}{y} - \frac{dx}{x} = 
    \begin{cases}
        (1 + \gamma \frac{1-w}{w}) \frac{dy}{y} &\text{ if } dy>0, \\
        (1 + \frac{1}{\gamma} \frac{1-w}{w}) \frac{dy}{y} &\text{ if } dy<0.
    \end{cases}
\end{equation}
This equation shows how buying pressure ($dy>0$) pushes the price up while selling pressure ($dy<0$) pushes it down. The transaction cost $\gamma$ affects the magnitude of these price changes.

\subsubsection{Liquidity Dynamics}
The liquidity $\ell$ of the pool changes as follows:
\begin{equation} \label{eqn:liquidity_diff}
    \frac{d\ell}{\ell} = w \frac{dx}{x} + (1-w) \frac{dy}{y} = 
    \begin{cases}
    (1-\gamma) (1-w) \frac{dy}{y} &\text{ if } dy>0, \\
    (1 - \frac{1}{\gamma}) (1-w) \frac{dy}{y} &\text{ if } dy<0.
    \end{cases}
\end{equation}
This equation reveals that the liquidity consistently increases due to the transaction costs collected by the pool.

\begin{remark}
In the continuous trading regime, the G3M effectively maintains certain quantities constant, depending on the direction of trading:
\begin{itemize}
    \item When buying $X$: $x^w y^{\gamma (1-w)}$ remains constant.
    \item When selling $X$: $x^{\gamma w} y^{1-w}$ remains constant.
\end{itemize}
This behavior leads to a continuous accumulation of liquidity in the pool.
\end{remark}

\subsection{LP Wealth}
The wealth of a liquidity provider (LP) in a G3M is determined by the value of their holdings in the pool. At any given time, the LP's wealth, denoted by $V(P)$, is simply the sum of the value of their $X$ holdings and their $Y$ holdings:
$
V(P) = Px + y
$
where $P$ is the current price of asset $X$ relative to asset $Y$.

Using the relationship between the reserves, price, and liquidity (equation \eqref{eqn:weight}), we can express the LP's wealth in a more convenient form:
\begin{equation} \label{eqn:LP_wealth}
V(P)
= \frac{Px}{w}
= \frac{Px}{w} \left(\frac{\frac{Px}{w}}{\frac{Px}{w}} \right)^w \left(\frac{\frac{y}{1-w}}{\frac{Px}{w}} \right)^{1-w}
= \frac{\ell P^w}{w^w (1-w)^{1-w}}.
\end{equation}
This equation shows that the LP's wealth is directly proportional to the pool's liquidity $\ell$ and depends on the price $P$ raised to the power of the weight $w$. Taking the logarithm of both sides, we obtain the log wealth:
\begin{equation} \label{eqn:log_LP}
    \ln V(P) = \ln \ell + w \ln P + \mathcal{S}_w,
\end{equation}
where $\mathcal{S}_w = - w \ln w - (1-w) \ln (1-w)$. The term $\mathcal{S}_w$ is the entropy of the weight distribution $(w,1-w)$, which reflects the diversification of the LP's holdings.

\begin{remark}
The LP's wealth is expressed here in terms of the G3M pool price $P$. This is natural from the LP's perspective, as they may not always have access to the true reference market price. Furthermore, this formulation (equation \eqref{eqn:log_LP}) is crucial for computing the long-term growth rate of LP wealth.

Section \ref{sec:LP_wealth_growth} compares this approach to a valuation based on the reference market price. This analysis will show that the logarithmic values of these two expressions differ by, at most, a constant factor, which depends on the transaction cost parameter $\gamma$.
\end{remark}
\section{Arbitrage-Driven G3M Dynamics} \label{sec:arbitrage_driven_dynamics}
This section investigates how the presence of arbitrageurs influences the behavior of a G3M. Arbitrageurs are traders who exploit price discrepancies between different markets to profit. In our context, they take advantage of any differences between the G3M pool price and the price on an external reference market.

To focus specifically on the impact of arbitrage, we make two simplifying assumptions:
\begin{assumption} \label{assumption:no_noisetrading}
There are no noise traders in the market.
\end{assumption}
Noise traders are those who trade based on non-fundamental factors, introducing randomness into the market. By excluding them, we can isolate the effects of arbitrageurs who act rationally based on price discrepancies.

\begin{assumption} \label{assumption:frictionless}
There exists an external reference market with infinite liquidity and no trading costs.
\end{assumption}
This assumption ensures that arbitrageurs can execute trades in the reference market without incurring any costs or affecting the market price. This idealized setting allows us to focus on the arbitrageurs' impact on the G3M.

\subsection{Arbitrage Bounds and Price Dynamics} \label{sec:arbitage_bound}
In the presence of a frictionless external reference market (Assumption \ref{assumption:frictionless}), arbitrageurs can freely exploit any price discrepancies between the G3M and the reference market. This arbitrage activity imposes bounds on the G3M price, preventing it from straying too far from the reference market price.

To understand these bounds, let's denote the price of asset $X$ relative to asset $Y$ in the reference market as $S$, and the corresponding price in the G3M pool as $P$. Arbitrageurs aim to maximize their profit by buying asset $X$ where it's cheaper and selling it where it's more expensive. This leads to the following optimization problem for an arbitrageur, as shown in \cite{angeris2019analysis}:
\begin{align} \label{eqn:arbitrage_optimization}
    \max_{\Delta \in \mathbb{R}^2} \, &- ( S\Delta_x + \Delta_y ) \\
\text{s.t. }
    &(x_0 + \gamma \Delta_x)^w (y_0 + \Delta_y)^{1-w} = x_0^w y_0^{1-w} \quad \text{if } \Delta_x \geq 0 \notag \\
    &(x_0 + \Delta_x)^w (y_0 + \gamma \Delta_y)^{1-w} = x_0^w y_0^{1-w} \quad \text{ if } \Delta_x < 0, \notag
\end{align}
where $(x_0, y_0)$ represents the initial reserve of assets $X$ and $Y$ in the pool. This optimization problem captures the arbitrageur's objective: maximize profit while adhering to the G3M's trading rules (equations \eqref{eqn:AMM_buy} and \eqref{eqn:AMM_sell}), which include transaction costs.

The solution to this optimization problem reveals that the G3M price, after arbitrageurs have acted, must satisfy
$$
P^* =
\begin{cases}
    \gamma^{-1}S &\text{if } P_0 > \gamma^{-1}S, \\
    P_0 &\text{if } \gamma S \leq P_0 \leq \gamma^{-1} S, \\
    \gamma S &\text{if } P_0 < \gamma S.
\end{cases}
$$
where $\gamma$ is the transaction cost parameter of the G3M \cite[\S 2.1]{angeris2019analysis}. This interval, $[\gamma S, \gamma^{-1}S]$, defines the \textit{no-arbitrage bounds}.

Essentially, these bounds create a "safe zone" for the G3M price. If the price falls outside this zone, arbitrageurs will immediately execute trades, buying or selling asset $X$ until the price returns within the bounds. This arbitrage activity effectively regulates the G3M price, keeping it anchored to the reference market price.

Figure \ref{fig:pool_CEX_timeseries} provides empirical evidence of this behavior. It shows a time series of the G3M price, clearly demonstrating that it consistently stays within the no-arbitrage bounds.

\begin{figure}[h!]
\centering
\includegraphics[width=1\linewidth]{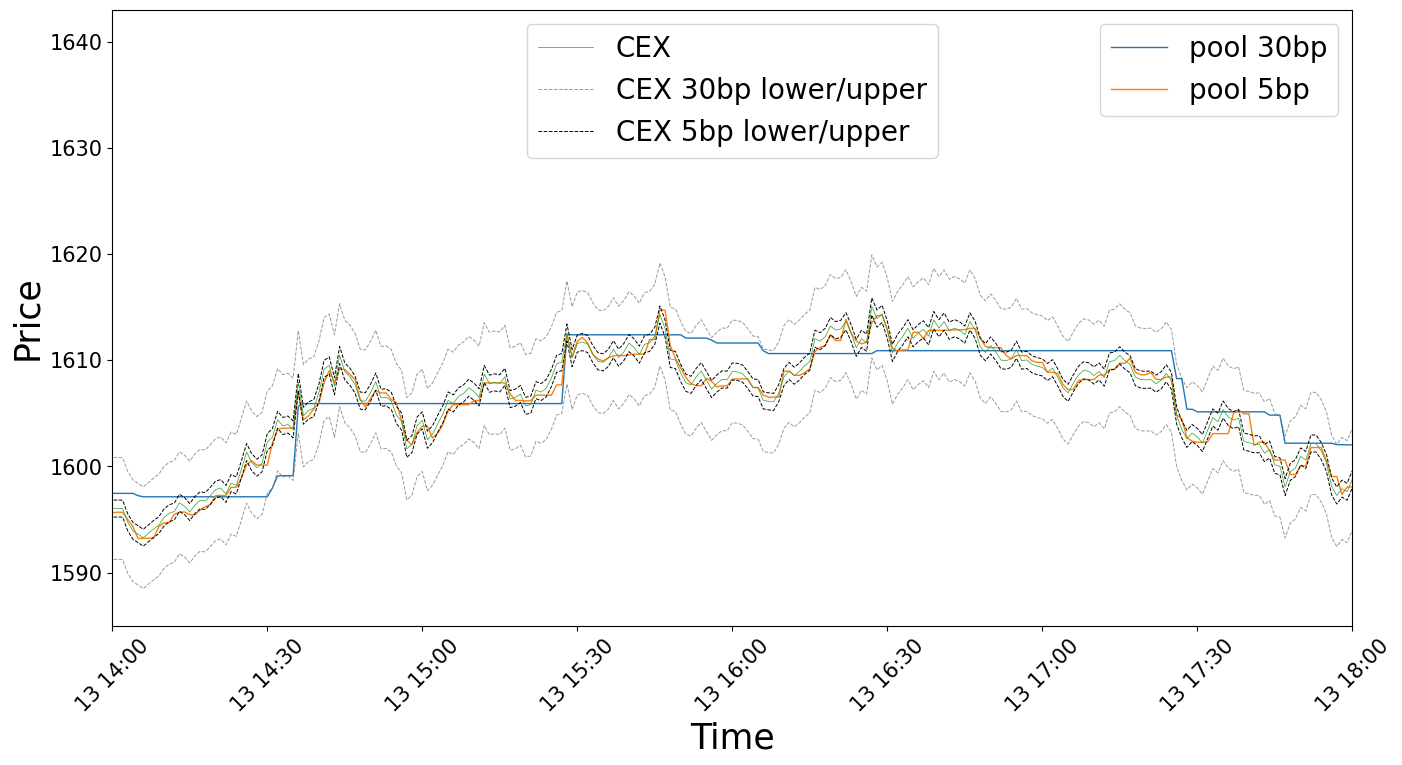}
\caption{Price time series (1-minute intervals) from 14:00:00 to 18:00:00 on September 13, 2023. Dashed lines indicate the upper and lower boundaries of the no-arbitrage region.}
\label{fig:pool_CEX_timeseries}
\end{figure}

\subsection{Mispricing and Arbitrage Dynamics}
To formally analyze how arbitrageurs influence the G3M's behavior, we introduce the concept of \textbf{mispricing}, which quantifies the discrepancy between the G3M pool price and the reference market price. We investigate how this mispricing evolves under both discrete and continuous arbitrage scenarios, building upon the framework in \cite{milionis2022automated,milionis2023automated}.

\subsubsection{Discrete Arbitrage Model}
First, we consider a discrete-time model where arbitrageurs arrive at discrete times:
\begin{assumption} \label{assumption:discrete_arbitrage}
    Arbitrageurs arrive at discrete times $0 = \tau_0 < \tau_1 < \tau_2 < \cdots \tau_m \leq T$.
\end{assumption}

At each arrival time $\tau_i$, an arbitrageur observes the reference market price $S_{\tau_i}$ and the prevailing G3M pool price $P_{\tau_{i-1}}$. They then execute trades to exploit any price difference, aiming to maximize their profit.

To quantify this price difference, we define the mispricing $Z_t$ as: 
\begin{equation} \label{eqn:Z}
    Z_t \triangleq \ln \frac{S_t}{P_t}.
\end{equation}
This quantity measures the logarithmic difference between the reference market price $S_t$ and the G3M pool price $P_t$.

Following the discussion of Section \ref{sec:arbitage_bound}, the arbitrage process can be described by the following recipe, similar to \cite[Lemma 2]{milionis2023automated}:
\begin{itemize}
    \item \textbf{If the G3M price is too low relative to the reference market} ($S_{\tau_i} < \gamma P_{\tau_{i-1}}$ or equivalently $Z_{\tau_i^-} < \ln \gamma$), the arbitrageur sells asset $X$ to the pool at the relatively higher price and simultaneously buys it on the reference market at the lower price. This pushes the G3M price up until it reaches the upper no-arbitrage bound ($Z_{\tau_i} = \ln \gamma$).
    \item \textbf{If the G3M price is too high relative to the reference market} ($S_{\tau_i} > \gamma^{-1} P_{\tau_{i-1}}$ or equivalently $Z_{\tau_i^-} > - \ln \gamma$), the arbitrageur buys asset $X$ from the pool at the relatively cheaper price and immediately sells it on the reference market at the higher price. This arbitrage activity pushes the G3M price down until it reaches the lower no-arbitrage bound ($Z_{\tau_i} = - \ln \gamma$).
    \item \textbf{If the G3M price is already within the no-arbitrage bounds} ($\gamma P_{\tau_{i-1}} \leq S_{\tau_i} \leq \gamma^{-1} P_{\tau_{i-1}}$ or equivalently $\ln \gamma \leq Z_{\tau_i^-} \leq -\ln \gamma$), there is no profitable arbitrage opportunity, and the arbitrageur does not execute any trades.
\end{itemize}
This arbitrage process leads to the following update rule for the G3M price at each arbitrageur arrival time:
\begin{equation} \label{eqn:P_update}
P_{\tau_i} = \begin{cases}
    \gamma S_{\tau_i} &\text{ if } Z_{\tau_i^-} < \ln \gamma, \\
    P_{\tau_{i-1}} &\text{ if } \ln \gamma \leq Z_{\tau_i^-} \leq -\ln \gamma, \\
    \gamma^{-1} S_{\tau_i} &\text{ if } Z_{\tau_i^-} > - \ln \gamma.
\end{cases}
\end{equation}
Accordingly, the mispricing process evolves as:
\begin{align} \label{eqn:Z_update}
Z_{\tau_i}
&= \max\left\{\min\{Z_{\tau_i^-}, -\ln \gamma\}, \ln \gamma \right\} \\
&= \begin{cases}
    \ln \gamma &\text{ if } Z_{\tau_i^-} < \ln \gamma, \\
    Z_{\tau_i^-}  &\text{ if } \ln \gamma \leq Z_{\tau_i^-} \leq -\ln \gamma, \\
    - \ln \gamma &\text{ if } Z_{\tau_i^-} > - \ln \gamma.
\end{cases} \notag
\end{align}
These equations capture how arbitrageurs adjust the G3M price in discrete steps to keep it within the no-arbitrage bounds.

\begin{proposition}[Discrete Mispricing Dynamics] \label{Prop:discrete_mispricing}
Given Assumptions \ref{assumption:no_noisetrading}, \ref{assumption:frictionless}, and \ref{assumption:discrete_arbitrage} and the initial condition $\gamma P_0 \le S_0 \le \gamma^{-1} P_0$, we can define:
\begin{equation*}
    J_i = \max\left\{\min\{Z_{\tau_i^-}, -\ln \gamma\}, \ln \gamma \right\} - Z_{\tau_i^-}, \quad
    L_t = \sum_{i: \tau_i \leq t} \{J_i\}^+, \quad
    U_t = \sum_{i: \tau_i \leq t} \{J_i\}^-,
\end{equation*}
where $\{a\}^+ = \max\{a, 0\}$ and $\{a\}^- = \max\{-a, 0\}$ denote the positive and negative parts of $a$, respectively. Then for all $t \geq 0$,
\begin{align*}
\ln P_t &= \ln P_0 + U_t - L_t, \\
Z_t &= \ln S_t - \ln P_0 + L_t - U_t.
\end{align*}
Moreover, $L_t$ and $U_t$ satisfy:
\begin{align}
L_t &= \sup_{i: \tau_i \leq t} \left( - \ln(\gamma P_0) + \ln S_{\tau_i} - U_{\tau_i} \right)^-, \label{eqn:sup_L} \\
U_t &= \sup_{i: \tau_i \leq t} \left( \ln(\gamma^{-1} P_0) - \ln S_{\tau_i} - L_{\tau_i} \right)^-. \label{eqn:sup_U}
\end{align}
\end{proposition}

\begin{proof}
The first assertion follows directly from \eqref{eqn:P_update} and \eqref{eqn:Z_update}. For the second assertion, note that both sides of \eqref{eqn:sup_L} and \eqref{eqn:sup_U} are non-decreasing and piecewise constant, with potential jumps occurring at $\tau_i$ for $0 \leq i \leq m$. It suffices to show the equalities at each $\tau_i$.

The proof employs induction on $i$. For $i=0$, the equalities hold by the initial condition $\ln \gamma \leq Z_0 \leq -\ln \gamma$. For $i<k$, the induction hypothesis yields: 
$$
L_{\tau_k} = \max \left\{ L_{\tau_{k-1}}, \left\{ -\ln \gamma + Z_{\tau_k} - L_{\tau_k} \right\}^- \right\}.
$$
Given that
$$
L_{\tau_k}
\begin{cases}
    \geq L_{\tau_{k-1}} &\text{ if } Z_{\tau_k} = \ln \gamma, \\
    = L_{\tau_{k-1}} &\text{ otherwise},
\end{cases}
$$
it follows that
$$
\left\{ -\ln \gamma + Z_{\tau_k} - L_{\tau_k} \right\}^-
= \min\{L_{\tau_k} + \ln \gamma - Z_{\tau_k}, 0 \}
\begin{cases}
    = L_{\tau_{k-1}} &\text{ if } Z_{\tau_k} = \ln \gamma, \\
    < L_{\tau_{k-1}} &\text{ otherwise},
\end{cases} 
$$
which validates \eqref{eqn:sup_L}. The same argument confirms \eqref{eqn:sup_U}.
\end{proof}

\begin{remark} \
\begin{itemize}
    \item The processes $L_t$ and $U_t$ act as \textbf{regulatory barriers} on the mispricing process $Z_t$. Specifically, at each arbitrage time $\tau_i$, $L_t$ prevents $Z_t$ from falling below the lower bound $\ln \gamma$, while $U_t$ prevents it from exceeding the upper bound $-\ln \gamma$. This ensures that the mispricing remains within the no-arbitrage interval $[\ln \gamma, -\ln \gamma]$.
    \item Fukasawa et al. \cite{fukasawa2023weighted, fukasawa2023modelfree} analyzed arbitrage in G3Ms by focusing on the dynamics of the reserve processes. Here, we adopt a different perspective by directly analyzing the mispricing process, which is compatible with their work through Equation \eqref{eqn:weight}. Our approach has the advantage of analyzing the growth rate of liquidity provider wealth, as discussed in Section \ref{sec:LP_growth_rate}. 
\end{itemize}
\end{remark}

\subsubsection{Continuous Arbitrage Model}
Now, let's consider the continuous-time limit, where arbitrageurs continuously monitor the market and react instantaneously to any arbitrage opportunities. This idealized scenario allows us to capture the dynamics of a highly efficient market with vigilant arbitrageurs.

Formally, we make the following assumption:
\begin{assumption} \label{assumption:conti_arbitrage}
Arbitrageurs continuously monitor and immediately act on arbitrage opportunities.
\end{assumption}

Under this assumption, arbitrageurs continuously adjust their trading strategies to maintain the G3M price within the no-arbitrage bounds. Their actions effectively prevent any significant mispricing from persisting.

The following proposition characterizes the dynamics of the mispricing process in this continuous setting:

\begin{proposition}[Continuous Mispricing Dynamics] \label{Prop:conti_mispricing}
Given that the market price $S_t$ is continuous and adheres to the initial condition $\gamma P_0 \le S_0 \le \gamma^{-1} P_0$, and under Assumptions \ref{assumption:no_noisetrading}, \ref{assumption:frictionless}, and \ref{assumption:conti_arbitrage}, we have:
\begin{enumerate}
\item[a)] The mispricing process, denoted as $Z_t$, can be decomposed into $Z_t = \ln S_t - \ln P_0 + L_t - U_t$ and takes value within the range $[\ln \gamma, -\ln \gamma]$ for all $t \ge 0$.
\item[b)] $L_t$ and $U_t$ are non-decreasing and continuous, with their initial values set at $L_0 = U_0 = 0$.
\item[c)] $L_t$ and $U_t$ increase only when $Z_t = - \ln \gamma$ and $Z_t = \ln \gamma$, respectively.
\end{enumerate}
Furthermore, $L_t$ and $U_t$ satisfy:
\begin{align}
L_t &= \sup_{0 \le s \le t} \left( - \ln(\gamma P_0) + \ln S_s - U_s \right)^-, \label{eqn:L} \\
U_t &= \sup_{0 \le s \le t} \left( \ln(\gamma^{-1} P_0) - \ln S_s - L_s \right)^-. \label{eqn:U}
\end{align}
\end{proposition}

\begin{proof}
Following the approach in \cite[\S 3]{fukasawa2023modelfree}, we decompose the log pool price as $\ln P_t = \ln P_0 + U_t - L_t$, where $U_t$ and $L_t$ represent the cumulative changes in log price due to buy and sell arbitrage orders, respectively. These processes naturally satisfy properties (a) to (c). Equations \eqref{eqn:L} and \eqref{eqn:U} then follow directly from the characterization of reflected processes in \cite[Proposition 2.4]{harrison2013brownian}.
\end{proof}

\begin{remark}
The mispricing process $Z_t$ can be viewed as a stochastic storage system with finite capacity, where $L_t$ and $U_t$ act as reflecting barriers that keep $Z_t$ within the allowed range \cite[\S 2.3]{harrison2013brownian}. In the special case where the reference market price $S_t$ follows a geometric Brownian motion, $Z_t$ becomes a reflected Brownian motion \cite[\S 6]{harrison2013brownian}.
\end{remark}

\begin{definition}
We say time $t$ is a \emph{point of increase} (resp.\ \emph{point of decrease}) for $x_t$ if there exists $\delta > 0$ such that $x_{t-\delta_1} < x_{t+\delta_2}$ (resp.\ $x_{t-\delta_1} > x_{t+\delta_2}$) for all $\delta_1, \delta_2 \in (0, \delta]$.
We say that $x_t$ \emph{increases} (resp.\ \emph{decreases}) \emph{only when} $Z_t = a$ if, at every point of increase (resp.\ point of decrease) of $x_t$, we have $Z_t = a$.
\end{definition}

%We now consider the dynamics of the reserve process $x_t$ (or $y_t$). Time $t$ is a point of increase (or decrease) for $x_t$ if there exists $\delta > 0$ such that $x_{t-\delta_1} < x_{t+\delta_2}$ (or $x_{t-\delta_1} > x_{t+\delta_2}$) for all $\delta_1, \delta_2 \in (0, \delta]$. The reserve process $x_t$ increases (or decreases) only when $Z_t = a$ if, at every point of increase (or decrease) for $x_t$, $Z_t = a$. 

\begin{corollary}[Inventory Dynamics in Arbitrage] \label{Cor:arbitrage_inventory}
Under the assumptions of Proposition \ref{Prop:conti_mispricing}, the following hold: 
\begin{enumerate}
    \item[(a)] $x_t$ and $y_t$ are predictable processes.
    \item[(b)] $x_t$ increases only when $Z_t = \ln \gamma$ and decreases only when $Z_t = - \ln \gamma$; $y_t$ increases only when $Z_t = - \ln \gamma$ and decreases only when $Z_t = \ln \gamma$.
    \item[(c)] $\ln x_t$ and $\ln y_t$ are continuous and of bounded variation on bounded intervals of $[0, \infty)$.
    \item[(d)] The arbitrage inventory process can be characterized by
    \begin{align} \label{eqn:xy_UL}
    d \ln x_t
    &= \frac{1-w}{1-w + \gamma w} dL_t - \frac{\gamma(1-w)}{\gamma(1-w) + w} dU_t, \\
    d \ln y_t
    &= \frac{w}{\gamma(1-w) + w} dU_t - \frac{\gamma w}{1-w + \gamma w} dL_t. \notag
    \end{align}
\end{enumerate}
\end{corollary}

\begin{proof}
By Proposition \ref{Prop:conti_mispricing} (a)--(c), $U_t$, $L_t$, and $\ln P_t$ are continuous, predictable processes of bounded variation. The assertions follow directly from \eqref{eqn:price_diff}. 
\end{proof}

The dynamics of liquidity growth based on the mispricing process can be described by incorporating \eqref{eqn:xy_UL} into \eqref{eqn:liquidity_diff}.

\begin{corollary}[Liquidity Dynamics in Arbitrage] \label{Cor:arbitrage_liquidity}
Under the assumptions of Proposition \ref{Prop:conti_mispricing}, the liquidity process $\ell_t$ is nondecreasing and predictable. Its rate of change is given by:
\begin{equation} \label{eqn:dl/l}
d \ln \ell_t= \frac{(1-\gamma)w(1-w)}{1-w+\gamma w} dL_t + \frac{(1 - \gamma)w(1-w)}{\gamma(1-w) + w} dU_t.
\end{equation}
\end{corollary}

\begin{remark}
The liquidity growth term in \eqref{eqn:dl/l} corresponds to the excess growth rate (see Section \ref{sec:market_model}) in Stochastic Portfolio Theory \cite[\S 1.1]{fernholz2002stochastic} as $\gamma \to 1$. 
\end{remark}

\subsection{LP Wealth Growth} \label{sec:LP_wealth_growth}
In this section, we analyze how the wealth of liquidity providers (LPs) evolves in the G3M under arbitrage-driven dynamics. To do this, we express the LP's wealth in terms of the reference market prices and the mispricing process.

The LP's wealth $V_t$ is defined as the total value of their holdings in the pool, denominated in terms of the reference market prices $S^X_t$ and $S^Y_t$ for assets $X$ and $Y$, respectively. This can be written as:
\begin{equation} \label{eqn:wealth_decomp}
    V_t = S^X_t x_t + S^Y_t y_t = S^Y_t (S_t x_t + y_t),
\end{equation}
where $S_t = S^X_t / S^Y_t$ is the relative price of asset $X$ in the reference market.

Using the relationship between the G3M pool price $P_t$ and the reserves (equation \eqref{eqn:weight}), we can derive bounds on the term $S_t x_t + y_t$: 
\begin{align*}
S_t x_t + y_t
&\geq \gamma P_t x_t + y_t
= \gamma \left( P_t x_t + y_t \right) + (1-\gamma) y_t
= \left(1 - w(1-\gamma) \right) \left( P_t x_t + y_t \right); \\
S_t x_t + y_t
&\leq \frac1\gamma P_t x_t + y_t
= \frac1\gamma \left( P_t x_t + y_t \right) + \left(1- \frac1\gamma \right) y_t
= \left\{1 + w\left(\frac1\gamma -1\right) \right\} \left( P_t x_t + y_t \right).
\end{align*}
These inequalities show that the value of the LP's holdings, $S_t x_t + y_t$, is always within a certain factor of the value based on the G3M pool price, $P_t x_t + y_t$. To quantify this relationship, we define the ratio:
$$
d_t := \ln \frac{S_t x_t + y_t}{P_t x_t + y_t}.
$$
From the above inequalities, we see that $d_t$ is bounded:
$$
1 - w(1-\gamma) \leq d_t \leq 1 + w(\gamma^{-1} -1).
$$

Now, we can express the logarithmic wealth of the LP as:
\begin{align} \label{eqn:log_wealth}
    \ln V_t &= \ln \ell_t + w \ln S^X_t + (1-w) \ln S^Y_t + d_t - Z_t \notag \\
    &= \frac{(1-\gamma)w(1-w)}{1-w+\gamma w} L_t + \frac{(1 - \gamma)w(1-w)}{\gamma(1-w) + w} U_t \notag \\
    &\quad + w \ln S^X_t + (1-w) \ln S^Y_t + d_t - Z_t,
\end{align}
Note that $d_t$ and $Z_t$ are bounded. This decomposition is the key to understanding the LP wealth growth in the G3M under arbitrage.

\section{LP Growth Rate} \label{sec:LP_growth_rate}
This section develops a methodology for calculating the ergodic growth rates of LP wealth components, denoted by $L_t$, $U_t$, and $V_t$. We begin by establishing the stochastic market model that governs the price dynamics of the assets involved.

\subsection{Market Model Setup} \label{sec:market_model}
We consider a market with two assets, $X$ and $Y$, whose price dynamics are governed by the following SDEs:
\begin{align} \label{eqn:market_dynamics}
d \ln S^X_t &= \mu^X_t dt + \xi^{XX}_t dB^X_t + \xi^{XY}_t dB^Y_t,\\
d \ln S^Y_t &= \mu^Y_t dt + \xi^{YX}_t dB^X_t + \xi^{YY}_t dB^Y_t, \notag
\end{align}
where $\mu^X_t$, $\mu^Y_t$, $\xi^{XX}_t$, $\xi^{XY}_t$, $\xi^{YX}_t$, and $\xi^{YY}_t$ are measurable and adapted, and $B^X_t$ and $B^Y_t$ are independent Brownian motions. The corresponding covariance processes for $\ln S^X_t$ and $\ln S^Y_t$ are given by
\begin{align} \label{eqn:covariance_processes}
    \sigma^{XX}_t dt &= d\langle \ln S^X, \ln S^X \rangle_t = \left( \xi^{XX}_t \xi^{XX}_t + \xi^{XY}_t \xi^{XY}_t \right) dt, \notag \\
    \sigma^{XY}_t dt &= d\langle \ln S^X, \ln S^Y \rangle_t = \left( \xi^{XX}_t \xi^{YX}_t + \xi^{XY}_t \xi^{YY}_t \right) dt = d\langle \ln S^Y, \ln S^X \rangle_t = \sigma^{YX}_t dt, \\
    \sigma^{YY}_t dt &= d\langle \ln S^Y, \ln S^Y \rangle_t = \left( \xi^{YX}_t \xi^{YX}_t + \xi^{YY}_t \xi^{YY}_t \right) dt. \notag
\end{align}

In this framework, the AMM price, denoted by $P_t$, is driven by the relative price of the two assets, $S_t = \frac{S^X_t}{S^Y_t}$. Applying It\^o's lemma to $S_t$, we obtain
\begin{align} \label{eqn:reference_SDE}
d \ln S_t
&= d \ln S^X_t - d \ln S^Y_t \notag \\
&= (\mu^X_t - \mu^Y_t) dt + (\xi^{XX}_t - \xi^{YX}_t) dB^X_t + (\xi^{XY}_t - \xi^{YY}_t) dB^Y_t \\
&:= \mu_t dt + \sigma_t dB^{XY}_t, \notag
\end{align}
where $\mu_t = \mu^X_t - \mu^Y_t$, $dB^{XY}_t = \frac{(\sigma^{XX}_t)^{\frac12}}{\sigma_t} dB^X_t - \frac{(\sigma^{YY}_t)^{\frac12}}{\sigma_t} dB^Y_t$, and $\sigma_t = \sqrt{\sigma^{XX}_t + \sigma^{YY}_t - 2 \sigma^{XY}_t}$.

Finally, we assume the following long-term limits exist:
\begin{align}\label{D:mu}
\mu_X := \lim_{T \to \infty} \frac{\Eof{\ln S^X_T}}{T}, \quad
\mu_Y := \lim_{T \to \infty} \frac{\Eof{\ln S^Y_T}}{T}.
\end{align}
These limits represent the long-term average growth rates of the asset prices. For convenience, we define $\mu = \mu_X - \mu_Y$, which represents the long-term average growth rate of the relative price $S_t$.

\subsection{Mispricing Process as a Reflected Diffusion}
To analyze the dynamics of mispricing, we model the log relative price $s_t := \ln S_t$ as a diffusion process:
\begin{equation} \label{eqn:log_reference_SDE}
    ds_t = \tilde{\mu}(t,s_t) dt + \tilde{\sigma}(t,s_t) dW_t,
\end{equation}
where $\tilde{\mu}(t,x)$ and $\tilde{\sigma}(t,x)$ both satisfy the standard Lipschitz and linear growth condition: for any $T>0$,
\begin{equation*}
    \sup_{t \in [0,T]} \sup_{x,y \in \mathbb{R}}\left[ \frac{|\tilde\mu(t,x)-\tilde\mu(t,y)| + |\tilde\sigma(t,x)-\tilde\sigma(t,y)|}{|x-y|} + \frac{|\tilde\mu(t,x)|+|\tilde\sigma(t,x)|}{1+|x|} \right] < \infty.
\end{equation*}
Fix a fee tier $\gamma \in (0,1)$ and set $c = -\ln \gamma$.
Suppose that $\ln(S_0/P_0) = z_0 \in (-c,c)$.
According to \cite[Proposition 2.4]{harrison2013brownian}, $s_t$ uniquely determines a process $Z_t$ that satisfies the following properties:
\begin{itemize}
    \item $Z_t$ takes values in $[-c,c]$;
    \item $Z_t = -\ln P_0 + s_t +L_t - U_t$; and
    \item $L_t$ and $U_t$ are two non-decreasing, continuous processes with $L_0=U_0=0$ that increase only when $Z_t=-c$ and $Z_t=c$, respectively.
\end{itemize}
Then, by uniqueness, $Z_t$ coincides with the mispricing process in Proposition \ref{Prop:conti_mispricing}.
In particular, the mispricing process $Z_t$ evolves according to
\begin{equation}\label{eqn:z-sde}
    \begin{cases}
        dZ_t = ds_t +dL_t - dU_t = \tilde{\mu}(t,s_t) dt + \tilde{\sigma}(t,s_t) dW_t + dL_t - dU_t,\\
        Z_0 = z_0 \in [-c,c].
    \end{cases}
\end{equation}
This shows that $Z_t$ is a diffusion process confined to the interval $[-c, c]$. The "reflection" at the boundaries captures the dynamics of mispricing being bounded within a certain range.
Throughout, we assume that $Z_t$ follows \eqref{eqn:z-sde} unless specified otherwise.

One would anticipate that $L_t$ and $U_t$ are boundary local times of $Z_t$. This is justified by the following lemma. For the basic theory on local times, see \cite{revuz1999continuous}.

\begin{lemma}\label{lem:local_time}
    For any $t>0$,
    \begin{align}
        &L_t = \lim_{\varepsilon \to 0^+} \frac{1}{2\varepsilon} \int_0^t {\bf 1}_{[-c,-c+\varepsilon)}(Z_r) \,d\langle Z\rangle_r,
        \label{eqn:L_t}\\
        &U_t = \lim_{\varepsilon \to 0^+}\frac{1}{2\varepsilon}\int_0^t {\bf 1}_{(c-\varepsilon,c]}(Z_r) \,d\langle Z\rangle_r,
        \label{eqn:U_t}
    \end{align}
    where $d\langle Z\rangle_r = \tilde\sigma^2(r,s_r) dr$ and the limits exist in $L^2$.
\end{lemma}

\begin{proof}
Let $f_\varepsilon(x)$ be the function such that $f_\varepsilon''(x) = \frac{1}{\varepsilon}{\bf 1}_{[-c,-c+\varepsilon)}(x)$ and $f_\varepsilon(-c)=f_\varepsilon'(-c)=0$.
By It\^o's lemma and the property that $L_t$ and $U_t$ increase only when $Z_t=-c$ and $Z_t=c$ respectively,
\begin{align*}
    df_\varepsilon(Z_t) &= df_\varepsilon'(Z_t) ds_t + f_\varepsilon'(-c) dL_t - f_\varepsilon'(c) dU_t + \frac{1}{2\varepsilon}{\bf 1}_{[-c,-c+\varepsilon)}(Z_t)\, d\langle Z\rangle_t.
\end{align*}
Since $f_\varepsilon'(-c)=0$ for all $\varepsilon>0$ and $f_\varepsilon'(c)=1$ for small enough $\varepsilon$, it follows that
\begin{align*}
    f_\varepsilon(Z_t) - f_\varepsilon(Z_0) = \int_0^t f_\varepsilon'(Z_r) ds_r - U_t + \frac{1}{2\varepsilon}\int_0^t {\bf 1}_{[-c,-c+\varepsilon)}(Z_r) \, d\langle Z \rangle_r.
\end{align*}
As $\varepsilon \to 0^+$, $f_\varepsilon(x) \to (x+c)^+$ and $f_\varepsilon'(x) \to {\bf 1}_{[-c, \infty)}(x)$. Since $Z_t \in [-c,c]$, we deduce that
\begin{align*}
    (Z_t+c) - (Z_0+c) = \int_0^t 1\, ds_r - U_t + \lim_{\varepsilon\to0^+}\frac{1}{2\varepsilon}\int_0^t {\bf 1}_{[-c,-c+\varepsilon)}(Z_r) \, d\langle Z \rangle_r,
\end{align*}
where the limit exists in $L^2$ thanks to It\^o's isometry.
This and the relation $dZ_t = ds_t + dL_t - dU_t$ from \eqref{eqn:z-sde} yield \eqref{eqn:L_t}.
Similarly, \eqref{eqn:U_t} can be obtained by considering the function $f_\varepsilon(x)$ such that $f_\varepsilon''(x)=\frac{1}{\varepsilon}{\bf 1}_{(c-\varepsilon,c]}(x)$ and $f_\varepsilon(c) = f_\varepsilon'(c) = 0$.
\end{proof}

Motivated by the study of the growth rate of LP's wealth, we examine quantities of the form
\begin{align}\label{quantity:g:int-f}
	g(Z_T) + \int_t^T f(r, Z_r) dr - \int_t^T \alpha_r dL_r + \int_t^T \beta_r dU_r.
\end{align}
The computation of the expectation of \eqref{quantity:g:int-f} simplifies when $\tilde\mu$ and $\tilde\sigma$ are constant because $Z_t$ has Markov property.
However, this is not the case when $\tilde\mu$ or $\tilde\sigma$ depend on $s_t$.
In order to overcome this technical obstacle, we appeal to the method of mimicking \cite{krylov1985once, gyongy1986mimicking} and adapt it to our current setting where local times are involved in the SDE.
Define
\begin{eqnarray} \label{eqn:mu_z-sig_z}
\mu(t, Z_t) = \Eof{\tilde{\mu}(t,s_t)|Z_t}, \qquad \sigma(t, Z_t) = \sqrt{\Eof{\tilde{\sigma}^2(t,s_t)|Z_t}}.
\end{eqnarray}
Consider the following SDE with reflection:
\begin{equation}\label{eqn:zhat-sde}
    \begin{cases}
        d\hat{Z}_t = \mu(t,\hat{Z}_t) dt + \sigma(t,\hat{Z}_t) dW_t + d\hat{L}_t - d\hat{U}_t,\\
        \hat{Z}_0 = z_0 \in (-c,c), \quad \hat{Z}_t \in [-c,c],
    \end{cases}
\end{equation}
where $\hat{L}_t$ and $\hat{U}_t$ are boundary local times of $\hat{Z}_t$ in the sense of Lemma \ref{lem:local_time}, i.e.,
\begin{align}
    &\hat{L}_t = \lim_{\varepsilon \to 0^+} \frac{1}{2\varepsilon} \int_0^t {\bf 1}_{[-c,-c+\varepsilon)}(\hat{Z}_r) \,d\langle \hat{Z}\rangle_r
    ,\label{eqn:L_t-hat}\\
    &\hat{U}_t = \lim_{\varepsilon \to 0^+}\frac{1}{2\varepsilon}\int_0^t {\bf 1}_{(c-\varepsilon,c]}(\hat{Z}_r) \,d\langle \hat{Z}\rangle_r.\label{eqn:U_t-hat}
\end{align}

\begin{lemma}
    The SDE \eqref{eqn:zhat-sde} has a unique strong solution $\hat{Z}_t$ with continuous sample paths and strong Markov property.
\end{lemma}

\begin{proof}
By Theorem 3.1 and Remark 3.3 of \cite{lions1984stochastic}, the Skorokhod problem
\begin{align}\label{skorokhod}
    \begin{cases}
        d\hat{Z}_t = \mu(t,\hat{Z}_t) dt + \sigma(t,\hat{Z}_t) dW_t - d k_t,\\
        \hat{Z}_0 = z_0 \in (-c,c),\\
        |k|_t = \int_0^t {\bf 1}_{\{\hat{Z}_r \in \partial[-c,c]\}} d|k|_r, \quad k_t = \int_0^t n(\hat{Z}_r) d|k|_r,
    \end{cases}
\end{align}
where $|k|_t$ denotes the total variation of $k_t$ and $n(x)$ denotes the outward unit normal for $x \in \partial[-c,c]$, has a unique solution $(\hat{Z}_t, k_t)$ where $\hat{Z}_t$ is a $[-c,c]$-valued continuous process and $k_t$ is a nonnegative continuous bounded variation process.
Since $\partial[-c,c] = \{-c,c\}$, $n(-c) = -1$, and $n(c) = 1$, we can write
\begin{align*}
    |k|_t = k^{-c}_t + k^c_t \quad \text{and} \quad
    k_t = k^c_t - k^{-c}_t,
\end{align*}
where $k^{-c}_t$ and $k^c_t$ are two non-decreasing continuous processes with $k^{-c}_0 = k^c_0 = 0$ that increase only when $\hat{Z}_t = -c$ and $\hat{Z}_t = c$, respectively. In particular,
\begin{align*}
    d\hat{Z}_t = \mu(t,\hat{Z}_t) dt + \sigma(t,\hat{Z}_t) dW_t + d k^{-c}_t - d k^c_t.
\end{align*}
Then, we may use It\^o's lemma as in the proof of Lemma \ref{lem:local_time} to show that $k^{-c}_t=\hat{L}_t$ and $k^c_t=\hat{U}_t$, and hence \eqref{eqn:zhat-sde} holds.
Finally, as we know the solution $\hat{Z}_t$ to \eqref{eqn:zhat-sde} is pathwise unique (because if $\tilde{Z}_t$ is another solution, then $(\tilde{Z}_t, \tilde{U}_t-\tilde{L}_t)$ with $\tilde{L}_t$ and $\tilde{U}_t$ being boundary local times of $\tilde{Z}_t$ is a solution to \eqref{skorokhod} which is known to be unique), it follows from the standard theory that the solution $\hat{Z}_t$ to \eqref{eqn:zhat-sde} is weakly unique and hence has strong Markov property; see \cite[\S I.4--I.5]{bass1998diffusions} or \cite[\S 5.3--5.4]{karatzas1991brownian}.
\end{proof}

In order to avoid technicalities and focus on the main ideas,
%carry out further analysis, 
we impose hereafter the following assumptions:
\begin{enumerate}
    \item Strong ellipticity condition:
    \begin{align}\label{strong_ellipticity}
        \text{for any $T>0$, } \inf_{t \in [0,T]} \inf_{z \in [-c,c]} \sigma^2(t,z) \ge \varepsilon > 0;
    \end{align}
    \item Existence of marginal densities:
    \begin{align}\label{density}
        \text{for each $t>0$, $Z_t$ and $\hat{Z}_t$ have densities $p(t,z)$ and $\hat{p}(t,z)$, respectively.}
    \end{align}
    \item Regularities: $\mu$ is $C^1$ in $z$, $\sigma$ is $C^2$ in $z$, and $p, \hat{p} \in C^{1,2}((0,T]\times[-c,c])$.
\end{enumerate}

\begin{lemma}[Mimicking]\label{lem:z=zhat}
    For any $t>0$, the random variables $Z_t$ and $\hat{Z}_t$ share the same marginal distribution.
\end{lemma}

\begin{proof}
By It\^o's lemma, for any test function $f \in C^\infty([-c,c])$,
\begin{align}\begin{split}\label{eqn:ito:zhat}
    &\mathbb{E}\left[f(\hat{Z}_t)\right]-f(z_0)\\
    &= \mathbb{E}\left[ \int_0^t \left( f_z(\hat{Z}_r) \mu(r,\hat{Z}_r) + \frac12 f_{zz}(\hat{Z}_r) \sigma^2(r,\hat{Z}_r) \right) dr + \int_0^t f_z(\hat{Z}_r) d\hat{L}_r - \int_0^t f_z(\hat{Z}_r) d\hat{U}_r \right].
\end{split}\end{align}
Under \eqref{density}, it follows that for any smooth test function $f$ vanishing near $\pm c$,
\begin{align*}
    \int_{-c}^c f(z) \hat{p}(t,z) dz - f(z_0)
    = \int_0^t dr \int_{-c}^c dz \left( f_z(z) \mu(r,z) + \frac12 f_{zz}(z) \sigma^2(r,z) \right) \hat{p}(r,z),
\end{align*}
where apparently $\hat p(t,z)$ denotes the probability density for $\hat Z_t$.
Applying $\partial_t$ and integration by parts, we deduce that
\begin{align*}
    \int_{-c}^c f(z) \partial_t \hat{p}(t,z) dz = \int_{-c}^c \left(-f(z)\partial_z(\mu(t,z)\hat{p}(t,z)) + \frac12 f(z) \partial_{zz}(\sigma^2(t,z) \hat{p}(t,z)) \right) dz.
\end{align*}
This shows that $\hat{p}(t,z)$ solves the Fokker-Planck equation
\begin{align}\label{eqn:phat}
    \partial_t \hat{p}(t,z) = -\partial_z(\mu(t,z)\hat{p}(t,z)) + \frac12 \partial_{zz}(\sigma^2(t,z) \hat{p}(t,z)), \quad t > 0, z \in (-c,c).
\end{align}
To derive the boundary conditions for $\hat{p}$, we first note that for any test function $f$ (but this time not necessarily vanishing at $\pm c$),
\begin{align*}
    \mathbb{E}\left[ \int_0^t f_z(\hat{Z}_r) d\hat{L}_t \right]
    &= \lim_{\varepsilon \to 0^+} \frac{1}{2\varepsilon} \mathbb{E}\left[\int_0^t  {\bf 1}_{[-c,-c+\varepsilon)}(\hat{Z}_r) f_z(\hat{Z}_r) \sigma^2(r,\hat{Z}_r) dr \right]\\
    &= \lim_{\varepsilon \to 0^+} \frac{1}{2\varepsilon} \int_0^t \int_{-c}^{-c+\varepsilon} f_z(z) \sigma^2(r,z) \hat{p}(r,z)\, dz \, dr \\
    & = \frac12 \int_0^t f_z(-c) \sigma^2(r,-c) \hat{p}(r,-c)\, dr
\end{align*}
and similarly,
\begin{align*}
    \mathbb{E}\left[ \int_0^t f_z(\hat{Z}_r) d\hat{U}_t \right]
    = \frac12 \int_0^t f_z(c) \sigma^2(r,c) \hat{p}(r,c)\, dr.
\end{align*}
Then, we may plug these into \eqref{eqn:ito:zhat}, apply $\partial_t$, and integrate by parts again to deduce that
\begin{align*}
    \int_{-c}^c f\partial_t \hat{p}\, dz &= \int_{-c}^c \left(f_z\mu\hat{p} + \frac12 f_{zz} \sigma^2 \hat{p}\right)dz - \left[f_z \sigma^2 p\right]_{-c}^c\\
    &= \int_{-c}^c f \left( -\partial_z(\mu \hat{p}) + \frac12 \partial_{zz}(\sigma^2 \hat{p}) \right) dz + \left[ f\mu \hat{p} - \frac12 f \partial_z(\sigma^2 \hat{p}) \right]_{-c}^c.
\end{align*}
Since $\hat{p}$ satisfies the PDE \eqref{eqn:phat}, this implies that $\hat{p}$ satisfies the boundary condition
\begin{align}\label{eqn:phat:bc}
    \mu(t,z) \hat{p}(t,z) - \frac12 \partial_z(\sigma^2(t,z)\hat{p}(t,z)) = 0 \qquad \text{at $z = \pm c$.}
\end{align}
On the other hand, we may apply It\^o's lemma to $f(Z_t)$ and use Fubini's theorem and the tower property to show that
\begin{align*}
    &\mathbb{E}\left[ f(Z_t)\right] - f(z_0)\\
    &=\mathbb{E}\left[ \int_0^t \left( f_z(Z_r)\tilde\mu(r,s_r) + \frac12 f_{zz}(Z_r) \tilde\sigma^2(r,s_r) \right) dr + \int_0^t f_z(Z_r) dL_r - \int_0^t f_z(Z_r) dU_r \right]\\
    &= \mathbb{E}\left[ \int_0^t \left( f_z(Z_r)\mathbb{E}[\tilde\mu(r,s_r)|Z_r] + \frac12 f_{zz}(Z_r) \mathbb{E}[\tilde\sigma^2(r,s_r)|Z_r] \right) dr + \int_0^t f_z(Z_r) dL_r - \int_0^t f_z(Z_r) dU_r \right]\\
    &=\mathbb{E}\left[ \int_0^t \left( f_z(Z_r)\mu(r,Z_r) + \frac12 f_{zz}(Z_r) \sigma^2(r,Z_r) \right) dr + \int_0^t f_z(Z_r) dL_r - \int_0^t f_z(Z_r) dU_r \right].
\end{align*}
With this, we repeat the same calculations above and deduce that $p(t,z)$ also solves the same PDE \eqref{eqn:phat} with the same boundary condition \eqref{eqn:phat:bc} and the same initial condition $\delta_{z_0}$.
Strong ellipticity \eqref{strong_ellipticity} ensures uniqueness of the solution (see Appendix \ref{append:uniqueness}), and hence $p(t,z) = \hat{p}(t,z)$ for all $t>0$ and $z \in [-c,c]$. This completes the proof that $Z_t$ and $\hat{Z}_t$ share the same marginal distribution.
\end{proof}

Even though the law of the mimicking process $\hat{Z}$ may not be the same as $Z$ (particularly when $\tilde\mu$ and $\tilde\sigma$ are non-constant), the next lemma shows that the computation of the expectation in \eqref{quantity:g:int-f} may be reduced to the computation of the same quantity for the mimicking process $\hat{Z}_t$.

\begin{lemma}\label{lem:ez=ezhat}
    For any given deterministic functions $\alpha_t$ and $\beta_t$ of $t$, the following equality holds.     \begin{align}\begin{split}\label{eqn:ez=ezhat}
        &\quad\ \Eof{g(Z_T) + \int_t^T f(r, Z_r) dr - \int_t^T \alpha_r dL_r + \int_t^T \beta_r dU_r}\\
        &= \Eof{g(\hat{Z}_T) + \int_t^T f(r, \hat{Z}_r) dr - \int_t^T \alpha_r d\hat{L}_r + \int_t^T \beta_r d\hat{U}_r}.
    \end{split}\end{align}
\end{lemma}

\begin{proof}
According to Lemma \ref{lem:z=zhat}, $Z_t$ and $\hat{Z}_t$ have the same marginal distribution, it follows that $\mathbb{E}[g(Z_T)] = \mathbb{E}[ g(\hat{Z}_T)]$ and
\begin{align*}
    \mathbb{E}\left[ \int_t^T f(r,Z_r)dr \right]
    = \int_t^T \mathbb{E}[ f(r,Z_r)] dr
    = \int_t^T \mathbb{E}[ f(r,\hat{Z}_r)] dr
    =\mathbb{E}\left[ \int_t^T f(r,\hat{Z}_r)dr \right]
\end{align*}
by applying Fubini's theorem. 
Moreover, thanks to Lemma \ref{lem:local_time} and the tower property,
\begin{align*}
    \mathbb{E}\left[ \int_t^T \alpha_r dL_r \right]
    &= \lim_{\varepsilon \to 0^+} \frac{1}{2\varepsilon} \int_t^T \alpha_r \, \mathbb{E}\left[ {\bf 1}_{[-c,-c+\varepsilon)}(Z_r) \, \tilde\sigma^2(r,s_r) \right]dr\\
    &= \lim_{\varepsilon \to 0^+} \frac{1}{2\varepsilon} \int_t^T \alpha_r \, \mathbb{E}\left[ {\bf 1}_{[-c,-c+\varepsilon)}(Z_r) \, \mathbb{E}\left[ \tilde\sigma^2(r,s_r) |Z_r \right] \right]dr\\
    &= \lim_{\varepsilon \to 0^+} \frac{1}{2\varepsilon} \int_t^T \alpha_r \, \mathbb{E}\left[ {\bf 1}_{[-c,-c+\varepsilon)}(Z_r) \, \sigma^2(r,Z_r) \right]dr \\
    &= \mathbb{E}\left[ \int_t^T \alpha_r d\hat{L}_r \right],
\end{align*}
where we have used \eqref{eqn:L_t-hat} in the last equality. 
The same applies to the integral for $dU_t$.
This concludes the proof.
\end{proof}

The following lemma yields a Feynman-Kac type formula, which provides a useful tool for analyzing the conditional expectation in \eqref{eqn:u-stoch-rep} through the PDE \eqref{eqn:z-pde}.

\begin{lemma}\label{lma:refl-diffusion-pde}
For any given deterministic functions $\alpha_t$ and $\beta_t$ of $t$, the solution $u(t,z)$ to the following parabolic PDE
\begin{equation} \label{eqn:z-pde}
u_t + \frac{\sigma^2(t,z)}2 u_{zz} + \mu(t,z) u_z + f(t,z) = 0
\end{equation}
with boundary condition
\[
    u_z(t, -c) = \alpha_t, \qquad u_z(t, c) = \beta_t
\]
and terminal condition $u(T, z) = g(z)$ has the following stochastic representation
\begin{equation} \label{eqn:u-stoch-rep}
u(t, z) = \mathbb{E}\,\left[g(\hat{Z}_T) + \int_t^T f(r, \hat{Z}_r) dr - \int_t^T \alpha_r d\hat{L}_r + \int_t^T \beta_r d\hat{U}_r\, \Big|\, \hat{Z}_t = z \right],
\end{equation}
where $\hat{Z}_t$ is the reflected diffusion given by \eqref{eqn:zhat-sde}.
\end{lemma}

\begin{proof}
It\^o's lemma implies that
\begin{align*}
& u(T, \hat{Z}_T) - u(t, \hat{Z}_t) \\
=& \int_t^T \left( u_t + \mu(r,\hat{Z}_r) u_z + \frac12 \sigma^2(r,\hat{Z}_r) u_{zz} \right)dr + \int_t^T \sigma(r,\hat{Z}_r) u_z dW_r + \int_t^T u_z d\hat{L}_r - \int_t^T u_z d\hat{U}_r,
\end{align*}
where $u_t$, $u_z$ and $u_{zz}$ are shorthands for $u_t(r,\hat{Z}_r)$, $u_z(r, \hat{Z}_r)$ and $u_{zz}(r, \hat{Z}_r)$, respectively.
Next, we take conditional expectation with respect to the filtration $\hat{\mathscr{F}}_t$ of $\hat{Z}_t$ on both sides of the last equation.
Using the boundary and terminal conditions of $u$ and the martingale property of $\int_0^t \sigma(r,\hat{Z}_r)u_z dW_r$, we deduce that
\begin{eqnarray*}
&& \mathbb{E}\left[g(\hat{Z}_T)\mid \hat{\mathscr{F}}_t \right] - u(t, \hat{Z}_t) \\
&=& \mathbb{E}\left[\int_t^T \left( u_t + \mu(r,\hat{Z}_r) u_z + \frac12 \sigma^2(r,\hat{Z}_r) u_{zz}\right) dr + \int_t^T u_z d\hat{L}_r - \int_t^T u_z d\hat{U}_r \, \Big| \, \hat{\mathscr{F}}_t\right]\\
&=& \mathbb{E}\left[-\int_t^T f(r, \hat{Z}_r)dr + \int_t^T \alpha_r d\hat{L}_r - \int_t^T \beta_r d\hat{U}_r \, \Big| \, \hat{\mathscr{F}}_t\right]
\end{eqnarray*}
since $u$ satisfies the PDE \eqref{eqn:z-pde}.
Thanks to the Markov property of $\hat{Z}_t$, we may replace $\hat{\mathscr{F}}_t$ by $\hat{Z}_t$ and rearrange terms to obtain the stochastic representation \eqref{eqn:u-stoch-rep}.
\end{proof}

In general, the solution to the terminal-boundary value problem in Lemma \ref{lma:refl-diffusion-pde} admits no simple analytical expression. We will focus on two special cases where the eigensystem associated with the operator in \eqref{eqn:u-pde} is readily accessible. To handle these cases, we introduce the following lemma, which addresses parabolic PDEs with non-zero Neumann boundary conditions.

\begin{lemma} \label{lma:u-v}
The solution $u$ to the parabolic PDE
\begin{equation} \label{eqn:u-pde}
u_t + \frac{\sigma^2(t,x)}2 u_{xx} + \mu(t,x) u_x = 0
\end{equation}
with boundary conditions $u_x(t, -c) = a$, $u_x(t, c) = b$ for some constants $a$, $b$, and terminal condition $u(T, x) = 0$ is given by
\begin{equation}
u(t, x) = v(t, x) + \frac{b-a}{4c} x^2 + \frac{b+a}2 x,
\end{equation}
where $v$ is the solution to the following inhomogeneous parabolic PDE
\[
v_t + \frac{\sigma^2(t,x)}2 v_{xx} + \mu(t,x) v_x + \frac{\sigma^2(x)}2 \frac{b-a}{2c} + \mu(x) \left( \frac{b - a}{2c} x + \frac{b + a}2\right) = 0
\]
with Neumann boundary conditions $v_x(t, -c) = v_x(t, c) = 0$ and terminal condition
\[
v(T, x) = -\frac{b-a}{4c} x^2 - \frac{b+a}2 x.
\]
\end{lemma}

\begin{proof}
The result follows from straightforward calculations.
\end{proof}

\subsection{Time-homogeneous reflected diffusion} \label{sec:time-homo}
This section focuses on the scenario where the reflected diffusion $\hat{Z}_t$ defined by \eqref{eqn:zhat-sde} is time-homogeneous. This means that the conditional drift $\mu$ and volatility $\sigma$ in \eqref{eqn:mu_z-sig_z} are independent of time. In this case, we can leverage the eigensystem of the infinitesimal generator of $\hat{Z}_t$ to simplify the analysis and express the conditional expectation in \eqref{eqn:u-stoch-rep}.

In the time-homogeneous case, the infinitesimal generator is $\cL = \frac{1}{2}\sigma^2(x) \p_x^2 + \mu(x) \p_x$ with Neumann boundary condition. In the Sturm-Liouville form (see \eqref{eqn:L-sl}),
$$
\cL = \frac1{\omega(x)} \frac{\p}{\p x} \left(p(x)\frac{\p}{\p x}\right),
$$
where
$$
p(x) := e^{\int \frac{2\mu(x)}{\sigma^2(x)} dx}, \qquad \omega(x) := \frac{2}{\sigma^2(x)} e^{\int \frac{2\mu(x)}{\sigma^2(x)} dx}.
$$
Here, $\omega$ represents the \textbf{speed measure} from classical diffusion theory, which measures how quickly the diffusion moves through different regions of the state space.

Let $\{(\lambda_n, e_n(x) )\}_{n=0}^\infty$ be the normalized eigensystem associated with $\cL$ with Neumann boundary condition in the interval $[-c,c]$. This eigensystem provides a convenient basis for analysis. The following theorem presents an eigensystem expansion for conditional expectations of integrals with respect to the boundary local times.

\begin{theorem}
For a time-homogeneous reflected diffusion $\hat{Z}_t$ defined by \eqref{eqn:zhat-sde} and any constants $a$ and $b$, the conditional expectation
\begin{equation} \label{eqn:cond-exp-LU}
\mathbb{E}\left[ - \int_t^T a d\hat{L}_\tau + \int_t^T b d\hat{U}_\tau \,\Big|\, \hat{Z}_t = x \right]
\end{equation}
admits the following eigensystem expansion associated with $\cL$ with Neumann boundary condition:
\begin{eqnarray}
&& \mathbb{E}\left[ - \int_t^T a d\hat{L}_\tau + \int_t^T b d\hat{U}_\tau \,\Big|\, \hat{Z}_t = x \right] \label{eqn:u-eig-expan} \\
&=& \frac{b-a}{4c} x^2 + \frac{b+a}2 x + \xi_0 (T - t) - \eta_0 + \sum_{n=1}^\infty \left\{\frac{\xi_n}{\lambda_n} \left[1 - e^{-\lambda_n(T - t)} \right] - \eta_n e^{-\lambda_n(T - t)} \right\} e_n(x), \nonumber
\end{eqnarray}
where the coefficients $\xi_n$ and $\eta_n$, for $n \geq 0$, are given by
\begin{equation} \label{eqn:xis-etas}
\xi_n = \int_{-c}^c h(x) e_n(x) \omega(x) dx, \qquad
\eta_n = \int_{-c}^c k(x) e_n(x) \omega(x) dx
\end{equation}
with
$$
h(x) = \frac{\sigma^2(x)}2 \frac{b-a}{2c} + \mu(x) \left( \frac{b - a}{2c} x + \frac{b + a}2\right), \qquad
k(x) = -\frac{b-a}{4c} x^2 - \frac{b+a}2 x.
$$
\end{theorem}

\begin{proof}
Let $u(t,z)$ denote the conditional expectation in \eqref{eqn:cond-exp-LU}.
Then, by Lemma \ref{lma:refl-diffusion-pde}, $u$ satisfies the PDE
\[
u_t + \frac{\sigma^2(x)}2 u_{xx} + \mu(x) u_x = 0
\]
with boundary condition $u_x(t,-c)=a$, $u_x(t,c)=b$, for all $t < T$, and terminal condition $u(T,x) = 0$. Thus, by Lemma \ref{lma:u-v}, the solution $u$ can be written as
\begin{equation*}
u(t, x) = v(t, x) + \frac{b-a}{4c} x^2 + \frac{b+a}2 x,
\end{equation*}
where $v$ is the solution to the following inhomogeneous parabolic PDE
\[
v_t + \frac{\sigma^2(x)}2 v_{xx} + \mu(x) v_x + \frac{\sigma^2(x)}2 \frac{b-a}{2c} + \mu(x) \left( \frac{b - a}{2c} x + \frac{b + a}2\right) = 0
\]
with Neumann boundary conditions $v_x(t, -c) = v_x(t, c) = 0$ and terminal condition
\[
v(T, x) = -\frac{b-a}{4c} x^2 - \frac{b+a}2 x.
\]
Hence, the eigensystem expansion for $v$ as in \eqref{eqn:u-eig-expan} in Section \ref{append:S-L} is given by
\[
v(t,x) = \xi_0 (T - t) - \eta_0
+ \sum_{n=1}^\infty \left\{\frac{\xi_n}{\lambda_n} \left[1 - e^{-\lambda_n(T - t)} \right] - \eta_n e^{-\lambda_n(T - t)} \right\} e_n(x),
\]
where the coefficients $\xi_n$'s and $\eta_n$'s are defined in \eqref{eqn:xis-etas}. Finally, since
\[
\mathbb{E}\left[- \int_t^T a d\hat{L}_\tau + \int_t^T b d\hat{U}_\tau \, \Big| \, \hat{Z}_t = x \right]  = u(t,x) = \frac{b-a}{4c} x^2 + \frac{b+a}2 x  + v(t,x),
\]
it follows that the eigensystem expansion \eqref{eqn:u-eig-expan} holds.
\end{proof}

This theorem provides a concrete way to calculate the conditional expectation in \eqref{eqn:cond-exp-LU} using the eigensystem of $\mathcal{L}$. This result enables us to analyze the long-term behavior of the mispricing process and its impact on LP's wealth.

\begin{corollary}\label{Cor:t-avg-lim-LU}
As $T \to \infty$, the limit of time-average of the expectation of \eqref{eqn:cond-exp-LU} is given by
\begin{align*}
\lim_{T\to\infty}& \frac1{T - t} \mathbb{E} \left[ -a (L_T - L_t) + b (U_T - U_t) \right] \\
&= \xi_0 = \frac{ \int_{-c}^c \left\{\frac{\sigma^2(x)}2 \frac{b-a}{2c} + \mu(x) \left( \frac{b-a}{2c} x + \frac{a+b}2\right)\right\} \omega(x) dx}{\int_{-c}^c \omega(x) dx}.
\end{align*}
\end{corollary}

\begin{proof}
Thanks to \eqref{eqn:u-eig-expan}, we have
\begin{align*}
    \lim_{T\to\infty}\frac{1}{T-t} \mathbb{E}\left[ -a(\hat{L}_T-\hat{L}_t) + b(\hat{U}_T-\hat{U}_t)  \mid \hat{Z}_t \right] = \xi_0 \quad \text{a.s.}
\end{align*}
We may apply expectation to the above and apply dominated convergence theorem to pass the limit under expectation, and then use Lemma \ref{lem:ez=ezhat} to finish the proof.
\end{proof}

This corollary, which can also be found in \cite{glynn2013central}, provides a concise expression for the long-term average behavior of the processes $L_t$ and $U_t$, which regulate the reflection of the mispricing process at the boundaries.

Leveraging this result, we can determine the long-term expected logarithmic growth rate of an LP's wealth in the G3M model.

\begin{theorem}[Log Growth Rate of LP Wealth in G3M Time-homogeneous] \label{thm:growth_rate}
Assuming that the functions $\mu$ and $\sigma$ defined in \eqref{eqn:mu_z-sig_z} are independent of $t$, the long-term expected logarithmic growth rate of an LP's wealth in a G3M is given by
\begin{equation} \label{eqn:growth_rate}
\lim_{T \to \infty} \frac{\mathbb{E}[\ln V_T]}{T}
= w \mu_X + (1-w) \mu_Y + \frac{(1-\gamma)w(1-w)}{1-w+\gamma w} \alpha + \frac{(1 - \gamma)w(1-w)}{\gamma(1-w) + w} \beta,
\end{equation}
where $\mu_X$ and $\mu_Y$ are the long-term growth rates of the asset prices defined in \eqref{D:mu}, and
\begin{align*}
&\alpha = \frac{\int_{\ln \gamma}^{-\ln \gamma} -\left\{\frac{\sigma^2(x)}{4 \ln \gamma} + \mu(x) \left( \frac{1}{2\ln \gamma} x - \frac12 \right)\right\} \omega(x) dx}{\int_{\ln \gamma}^{-\ln \gamma} \omega(x) dx}, \\
&\beta = \frac{\int_{\ln \gamma}^{-\ln \gamma} -\left\{\frac{\sigma^2(x)}{4 \ln \gamma} + \mu(x) \left( \frac{1}{2\ln \gamma} x + \frac12 \right)\right\} \omega(x) dx}{\int_{\ln \gamma}^{-\ln \gamma} \omega(x) dx}, \\
&\omega(x) = \frac{2}{\sigma^2(x)} e^{\int \frac{2\mu(x)}{\sigma^2(x)} dx}.
\end{align*}
\end{theorem}

\begin{proof}
This result follows directly from combining the expression for the logarithmic growth rate of LP's wealth \eqref{eqn:log_wealth}, the definitions of $\mu_X$ and $\mu_Y$ \eqref{D:mu}, and Corollary \ref{Cor:t-avg-lim-LU}.
\end{proof}

Theorem \ref{thm:growth_rate} provides a practical method for calculating the long-term growth rate of an LP's wealth. The terms $\alpha$ and $\beta$ quantify the influence of mispricing on this growth rate, capturing the effects of the boundaries on mispricing and the speed measure.

\begin{remark} \label{rmk:excess_growth_rate}
Theorem \ref{thm:growth_rate} reveals an intriguing connection between liquidity wealth dynamics in G3Ms under arbitrage-driven scenarios and the principles of Stochastic Portfolio Theory (SPT), particularly in frictionless markets. More precisely, Equation \eqref{eqn:growth_rate} corresponds to the growth rate
$$
w \mu_X + (1-w) \mu_Y + \frac{w(1-w)}2 \sigma^2
$$
of a constant rebalanced portfolio with the same weight, as discussed in \cite[\S 1.1]{fernholz2002stochastic}. Furthermore, the term $\frac{(1-\gamma)w(1-w)}{1-w+\gamma w} \alpha + \frac{(1 - \gamma)w(1-w)}{\gamma(1-w) + w} \beta$ converges to the excess growth rate $\frac{w(1-w)}2 \sigma^2$ as the fee tier $\gamma$ approaches 1.
\end{remark}

\subsection{Optimal fees and optimal growth}
To illustrate the link between LP wealth growth and SPT (as discussed in Remark \ref{rmk:excess_growth_rate}), we consider a G3M operating within a GBM market model. This simplifies the market dynamics in \eqref{eqn:market_dynamics} by assuming constant parameters for the asset prices. Consequently, the covariance processes in \eqref{eqn:covariance_processes} also have constant parameters, and the relative price $S_t = S^X_t / S^Y_t$ follows a simpler SDE:
$$
d \ln S_t = \mu dt + \sigma dB_t,
$$
where $\mu$ and $\sigma$ are constants defined in \eqref{eqn:reference_SDE}. This simplified setting allows us to explicitly compute the long-term expected logarithmic growth rate of LP wealth using Theorem \ref{thm:growth_rate}.

\begin{corollary}[Log Growth Rate of LP under GBM] \label{Cor:log_growth_GBM}
In a G3M operating under a GBM market with constant drift $\mu$ and volatility $\sigma$, the long-term logarithmic growth rate of LP wealth is:
\begin{align*}
\lim_{T \to \infty} \frac{\mathbb{E}[\ln V_T]}{T}
= w \mu_X + (1-w) \mu_Y + \frac{(1-\gamma)w(1-w)}{1-w+\gamma w} \alpha + \frac{(1-\gamma)w(1-w)}{\gamma(1-w)+ w} \beta,
\end{align*}
where $\theta = \frac{2 \mu}{\sigma^2}$ and
\begin{itemize}
\item $\alpha = \beta = -\frac{\sigma^2}{4 \ln \gamma}$ if $\theta = 0$.
\item $\alpha = \frac{\theta \sigma^2}{2(\gamma^{-2\theta} - 1)}$ and $\beta = \frac{\theta \sigma^2}{2(1-\gamma^{2\theta})}$ if $\theta \neq 0$.
\end{itemize}
Furthermore, the steady-state distribution $\pi$ of the mispricing process $Z_t$ is:
\begin{itemize}
    \item If $\theta = 0$, then $\pi$ is the uniform distribution on $[\ln \gamma, -\ln \gamma]$.
    \item If $\theta \neq 0$, $\pi$ is the truncated exponential distribution
    $$
    \pi(dz) = \frac{\theta e^{\theta z}}{\gamma^{-\theta} - \gamma^{\theta}} dz, \ z \in [\ln \gamma, -\ln \gamma].
    $$
\end{itemize}
\end{corollary}

\begin{proof}
Under the GBM assumption, the conditional expectations in \eqref{eqn:mu_z-sig_z} are constants and coincide with $\mu$ and $\sigma$.
In particular,
\begin{align*}
\omega(x) &=
\begin{cases}
    \frac{2}{\sigma^2} &\text{ if } \theta = 0, \\
    \frac{2}{\sigma^2} e^{\theta x} &\text{ if } \theta \neq 0,
\end{cases} \\
\int^{-\ln \gamma}_{\ln \gamma} \omega(x) dx &=
\begin{cases}
    \frac{-2 \ln \gamma}{\sigma^2} &\text{ if } \theta = 0, \\
    \frac{1}{\mu} \left[ \gamma^{-\theta} - \gamma^{\theta} \right] &\text{ if } \theta \neq 0,
\end{cases} \\
\int^{-\ln \gamma}_{\ln \gamma} x \omega(x) dx &=
\begin{cases}
    0 &\text{ if } \theta = 0, \\
    \frac{- \ln \gamma}{\mu} \left[\gamma^{-\theta} + \gamma^{\theta} \right] + \frac{1}{\theta\mu} \left[ \gamma^{-\theta} - \gamma^{\theta}  \right] &\text{ if } \theta \neq 0.
\end{cases}
\end{align*}
Hence, the first part of the corollary follows directly from Theorem \ref{thm:growth_rate}.
Furthermore, since $\mu$ and $\sigma$ are constant, the mispricing process $Z_t$ coincides with the reflected diffusion $\hat{Z}_t$ defined by \eqref{eqn:zhat-sde}.
Therefore, the second part of the corollary follows from Theorem \ref{thm:steady_state_RD} in the appendix.
\end{proof}

\begin{remark}
This result is consistent with the calculation in \cite[Proposition 6.6]{harrison2013brownian}, which utilizes the stationary distribution of the mispricing process.
\end{remark}

\subsubsection{Numerical Analysis and Optimal Fees}
We now numerically investigate the growth rate of LP wealth in G3Ms across different fee tiers $\gamma$ and asset weights $w$. This analysis extends the work in \cite[\S 5]{tassy2020growth}, which focused on the specific case of $w=\frac12$.

To quantify the impact of fees, we define $g(w, \gamma)$ as the ratio of the "mispricing-related" term in the G3M growth rate to the excess return in SPT. Corollary \ref{Cor:log_growth_GBM} yields
$$
g(w, \gamma) =
\begin{cases}
    - \frac{1}{2\ln \gamma} \left\{ \frac{1-\gamma}{1-w+\gamma w} + \frac{1-\gamma}{\gamma(1-w)+ w} \right\} \quad &\text{if } \theta=0; \\
    \theta \left\{ \frac{1-\gamma}{1-w+\gamma w} \frac{\gamma^{2\theta}}{1-\gamma^{2\theta}} + \frac{1-\gamma}{\gamma(1-w)+ w} \frac{1}{1-\gamma^{2\theta}} \right\} \quad &\text{if } \theta \neq 0.
\end{cases}
$$

Figures \ref{fig:growth_ratio_theta} and \ref{fig:growth_ratio_weight} illustrate this growth rate ratio for various weights and values of $\theta = \frac{2\mu}{\sigma^2}$. These figures reveal interesting "phase transitions" for both $\theta$ and $w$, where the ratio can exhibit non-monotonicity. This suggests that the optimal fee tier $\gamma^*$ that maximizes LP wealth growth may lie within the interior of the interval $(0,1)$, a behavior not observed when $w=\frac12$.

The figures also highlight the symmetry of $g(w,\gamma)$ with respect to $\theta$ and $w=\frac12$.
\begin{itemize}
    \item By replacing $\theta$ with $-\theta$, we obtain
    \begin{eqnarray*}
        && -\theta(1 - \gamma) \left\{ \frac1{1-w+\gamma w} \frac{\gamma^{-2\theta}}{1-\gamma^{-2\theta}} + \frac1{\gamma(1-w)+ w} \frac{1}{1-\gamma^{-2\theta}} \right\} \\
        &=& -\theta(1 - \gamma) \left\{ \frac1{1-w+\gamma w} \frac{1}{\gamma^{2\theta}-1} + \frac1{\gamma(1-w)+ w} \frac{\gamma^{2\theta}}{\gamma^{2\theta} - 1} \right\} \\
        &=& \theta(1 - \gamma) \left\{ \frac1{1-w+\gamma w} \frac{1}{1 - \gamma^{2\theta}} + \frac1{\gamma(1-w)+ w} \frac{\gamma^{2\theta}}{1 - \gamma^{2\theta}} \right\},
    \end{eqnarray*}
    which explains the symmetry between the plots for positive and negative values of $\theta$ in Figure \ref{fig:growth_ratio_theta}.
    \item Let $\delta = w - \frac12$. Then, $g$ can be expressed in terms of $\delta$ as
    $$
    g(\delta, \gamma) =
    \begin{cases}
        - \frac{1}{2\ln \gamma} \left\{ \frac{1-\gamma}{\frac12(1 + \gamma) + (\gamma - 1)\delta} + \frac{1-\gamma}{\frac12(1 + \gamma) - (\gamma - 1)\delta} \right\} \quad &\text{if } \theta=0; \\
        \theta \left\{ \frac{1-\gamma}{\frac12(1 + \gamma) + (\gamma - 1)\delta} \frac{\gamma^{2\theta}}{1-\gamma^{2\theta}} + \frac{1-\gamma}{\frac12(1 + \gamma) - (\gamma - 1)\delta} \frac{1}{1-\gamma^{2\theta}} \right\} \quad &\text{if } \theta \neq 0.
    \end{cases}
    $$
\end{itemize}
For $\theta = 0$, we have $g(-\delta,\gamma) = g(\delta,\gamma)$, explaining why only "red colors" appear when $\theta = 0$ in Figures \ref{fig:growth_ratio_theta} and \ref{fig:growth_ratio_weight}. For $\theta \neq 0$, a similar symmetry around $\delta = 0$ (i.e., $w = \frac12$) is observed, as shown in Figure \ref{fig:growth_ratio_theta}.

Heatmaps in Figures \ref{fig:growth_heatmap_weight} and \ref{fig:growth_heatmap_theta} further visualize the growth rate ratio. The location of the maximum value in each row is labeled, emphasizing the potential for non-monotonicity and interior optimal fee tiers. Notably, under certain conditions, the G3M can outperform both the unrebalanced ($\gamma^\ast = 0$) and constant rebalanced ($\gamma^\ast = 1$) portfolio strategies.

\begin{figure}[h!]
\centering
\includegraphics[width=\linewidth]{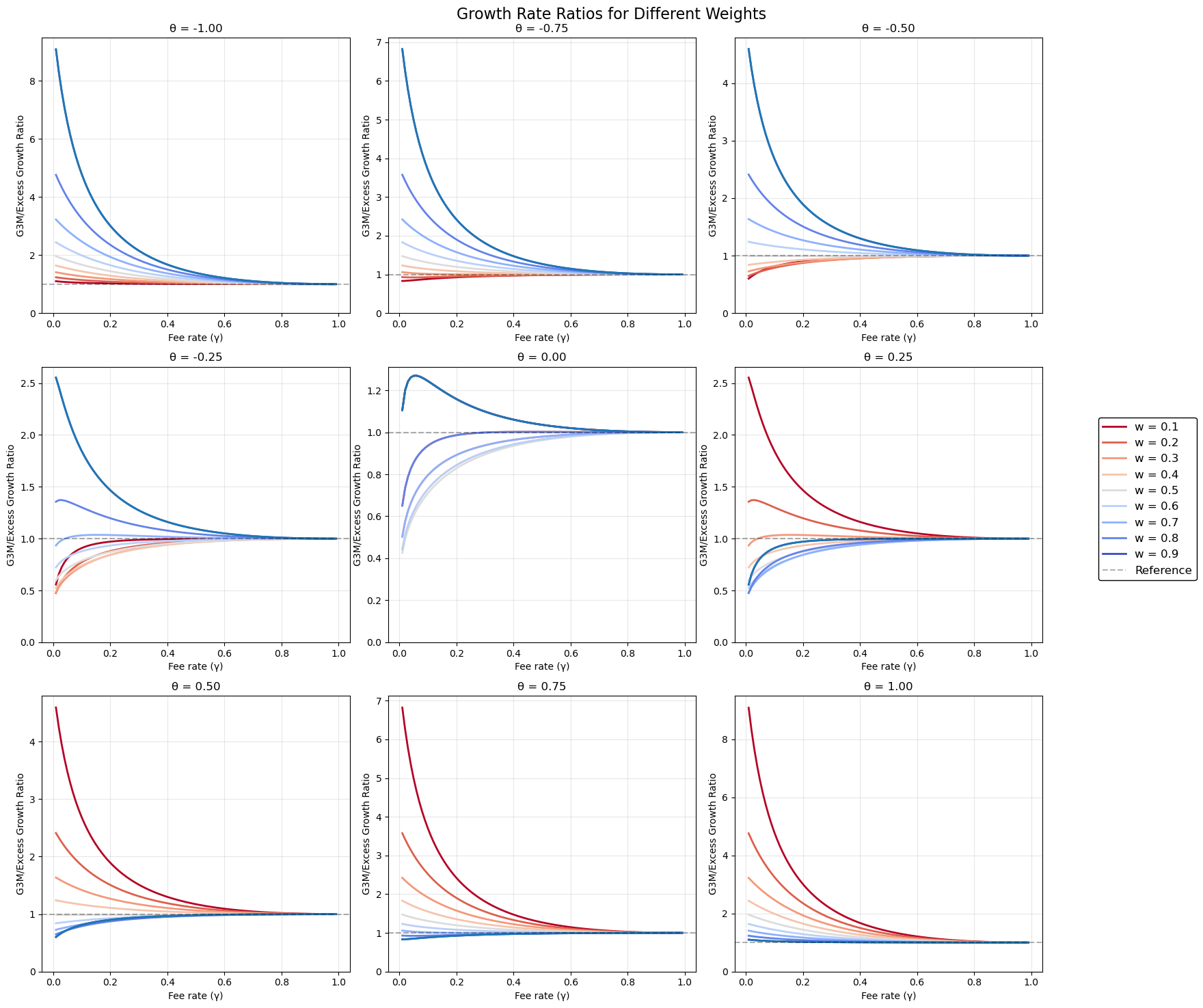}
\caption{Growth rate ratio for different values of weights. Note the symmetry between positive and negative $\theta$ values.}
\label{fig:growth_ratio_theta}
\end{figure}

\begin{figure}[h!]
\centering
\includegraphics[width=1\linewidth]{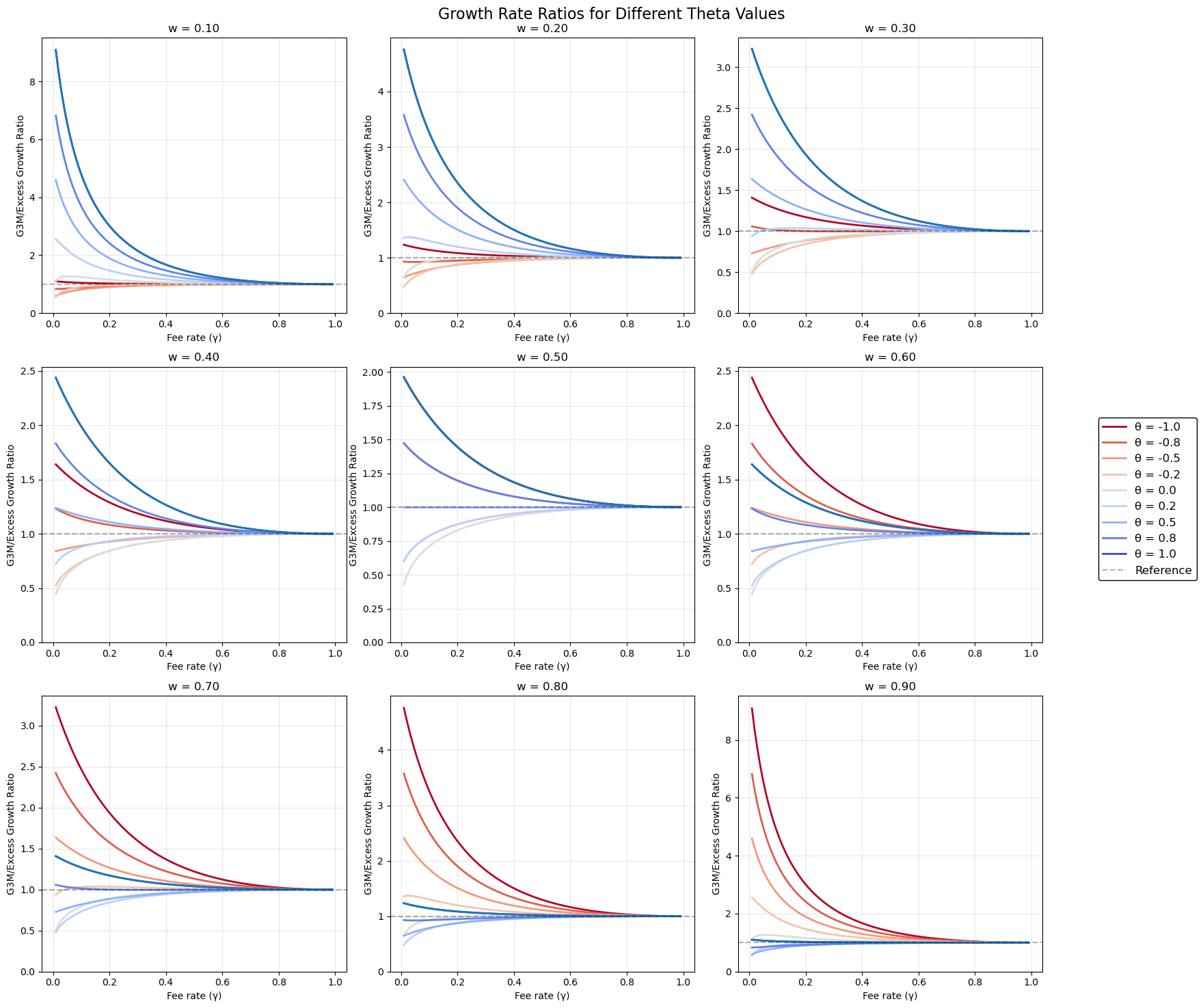}
\caption{Growth rate ratio for different values of $\theta$. Note the symmetry around $w = \frac12$.}
\label{fig:growth_ratio_weight}
\end{figure}

\begin{figure}[h!]
\centering
\includegraphics[width=1\linewidth]{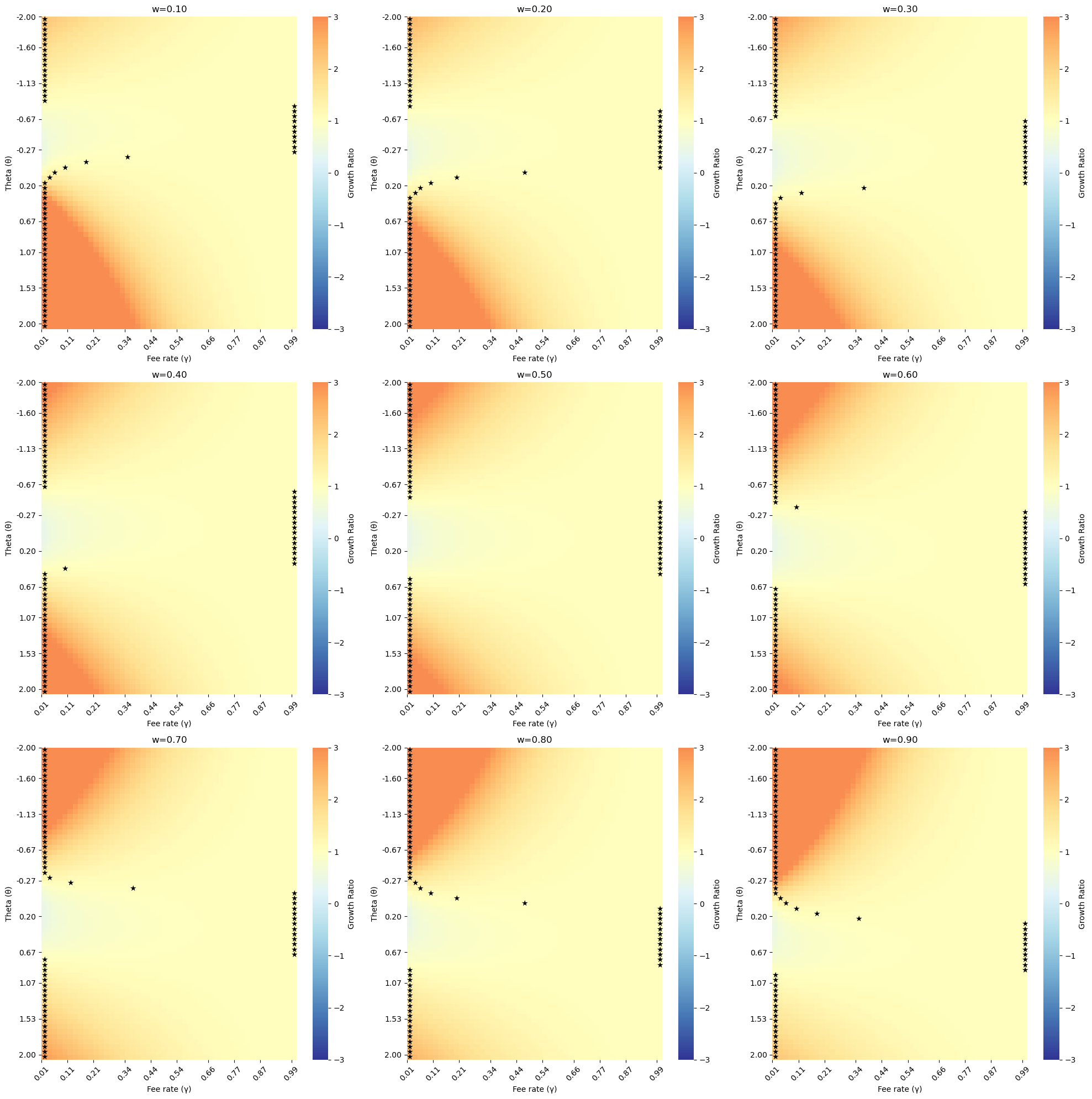}
\caption{Heatmap of growth rate ratio for different weights, with maximum values labeled.}
\label{fig:growth_heatmap_weight}
\end{figure}

\begin{figure}[h!]
\centering
\includegraphics[width=1\linewidth]{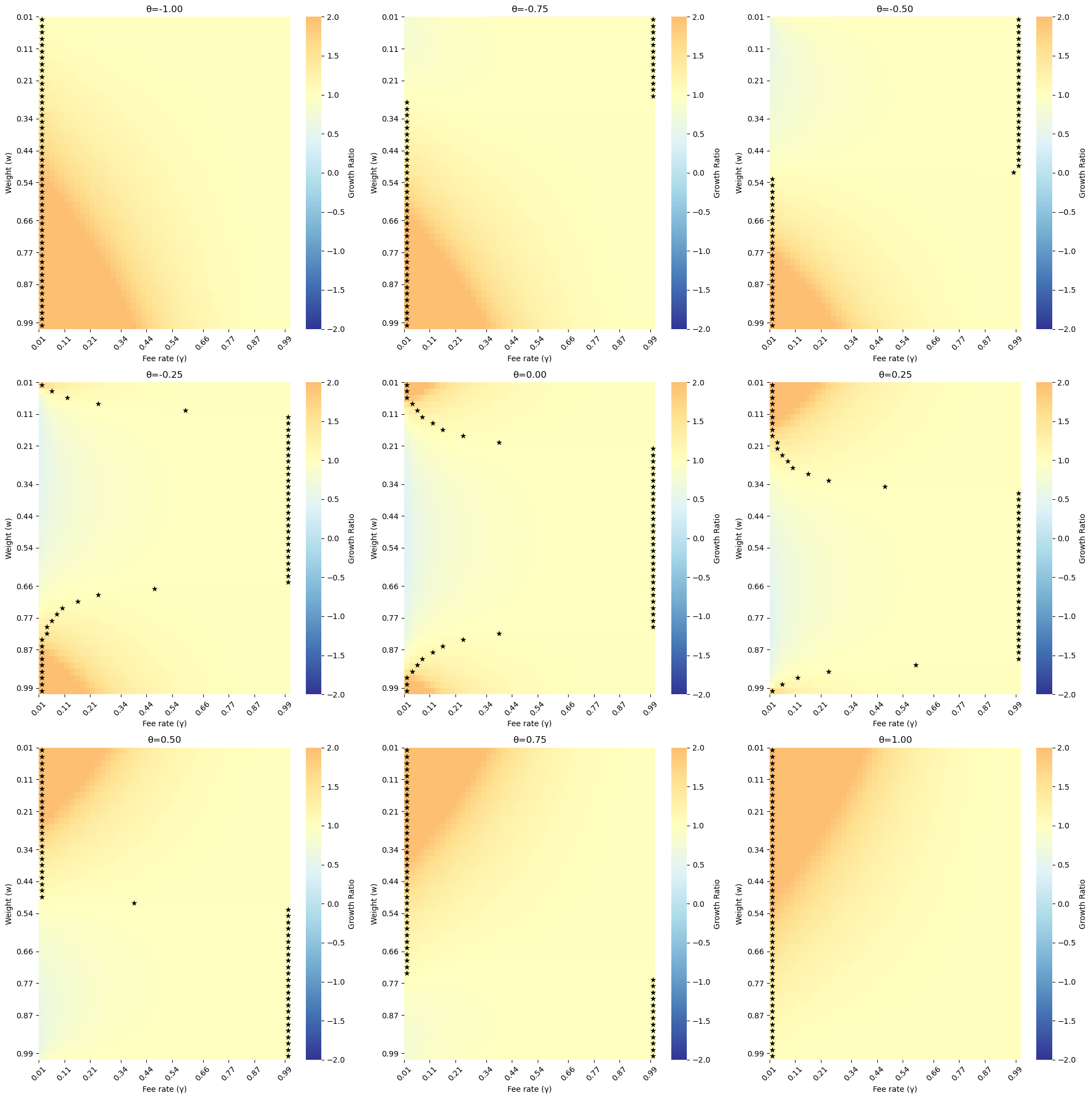}
\caption{Heatmap of growth rate ratio for different values of $\theta$, with maximum values labeled.}
\label{fig:growth_heatmap_theta}
\end{figure}

\subsection{Time-inhomogeneous reflected diffusion} \label{sec:time-inhomo}
In Section \ref{sec:time-homo}, we utilized the eigensystem of the infinitesimal generator to analyze the long-term growth rate of LP wealth for the time-homogeneous case. However, this approach is not applicable when the drift and volatility coefficients are time-dependent, as there are no universal eigenfunctions associated with time-varying eigenvalues.

Nevertheless, we can still determine the long-term expected logarithmic growth rate by analyzing the asymptotic behavior of the time-averaged expectation. This section presents a method to achieve this for time-inhomogeneous reflected diffusions.

Recall the mispricing process $Z_t$ governed by \eqref{eqn:z-sde}:
\begin{equation*}
dZ_t = \tilde{\sigma}(t, s_t) dW_t + \tilde{\mu}(t, s_t) dt + dL_t - dU_t.
\end{equation*}
We assume that the time-dependent drift $\mu(t, z)$ and volatility $\sigma(t,z)$ defined in \eqref{eqn:mu_z-sig_z} converge to limiting functions as $t \to \infty$:
\begin{equation*}
\lim_{t\to\infty} \sigma(t,z) = \sigma(z), \qquad \lim_{t\to\infty} \mu(t,z) = \mu(z)
\end{equation*}
in the $L^2$ sense, and that these limits are smooth and bounded. This implies that the mispricing process eventually approaches a time-homogeneous behavior. Define the limiting speed measure $q(y)$ as:
$$
q(y) = \frac{w(y)}{\int_{-c}^c w(\eta) d\eta}, \qquad
\mbox{ where }
w(y):= \frac2{\sigma^2(y)} e^{\int\frac{2\mu(y)}{\sigma^2(y)}dy}.
$$
This speed measure corresponds to the stationary distribution of the limiting time-homogeneous reflected diffusion with drift $\mu(x)$ and volatility $\sigma(x)$.

The following theorem, whose proof is deferred to Appendix \ref{append:time-inhomo}, characterizes the long-term time-averaged expectation for time-inhomogeneous reflected diffusions.

\begin{theorem} \label{thm:time-inhomo-lim}
Let $u=u(t,x)$ be the solution to the parabolic PDE
$$
u_t + \frac{\sigma^2(t,x)}2 u_{xx} + \mu(t,x) u_x + f(t,x) = 0
$$
with Neumann boundary condition $u_x(t,-c) = u_x(t,c) = 0$ and terminal condition $u(T,x) = g(x)$. Assume that $g\in L^2$ and there exists a function $\bar{f}(x)$ such that
$$
\lim_{t\to\infty} f(t,x) = \bar{f}(x)
$$
in $L^2$. Then, the following asymptotic behavior holds for $u(t,x)$ as $T \to \infty$:
$$
\lim_{T\to\infty} \frac{u(t,x)}{T-t} = \int_{-c}^c \bar{f}(y)q(y)dy.
$$
\end{theorem}

This theorem provides a tool for analyzing the long-term behavior of time-inhomogeneous reflected diffusions. As a direct consequence, we obtain the following result for the long-term growth rate of LP wealth in the G3M model.

\begin{corollary}[Log Growth Rate of LP Wealth in G3M Time-inhomogeneous] \label{thm:growth_rate_inhomo}
Assume that the mispricing process $Z_t$ is governed by \eqref{eqn:z-sde} with time-dependent coefficients $\sigma(t,z)$ and $\mu(t,z)$ defined in \eqref{eqn:mu_z-sig_z} converging to smooth and bounded limits $\sigma(z)$ and $\mu(z)$ in the $L^2$ sense as $t \to \infty$. Then, the long-term expected logarithmic growth rate of LP wealth in a G3M is:
\begin{align*}
\lim_{T \to \infty} \frac{\mathbb{E}[\ln V_T]}{T}
= w \mu_X + (1-w) \mu_Y + \frac{(1-\gamma)w(1-w)}{1-w+\gamma w} \alpha + \frac{(1 - \gamma)w(1-w)}{\gamma(1-w) + w} \beta,
\end{align*}
where
\begin{align*}
&\alpha = \frac{\int_{\ln \gamma}^{-\ln \gamma} -\left\{\frac{\sigma^2(x)}{4 \ln \gamma} + \mu(x) \left( \frac{1}{2\ln \gamma} x - \frac12 \right)\right\} \bar\omega(x) dx}{\int_{\ln \gamma}^{-\ln \gamma} \omega(x) dx}, \\
&\beta = \frac{\int_{\ln \gamma}^{-\ln \gamma} -\left\{\frac{\sigma^2(x)}{4 \ln \gamma} + \mu(x) \left( \frac{1}{2\ln \gamma} x + \frac12 \right)\right\} \omega(x) dx}{\int_{\ln \gamma}^{-\ln \gamma} \omega(x) dx}, \\
&\omega(x) = \frac{2}{\sigma^2(x)} e^{\int \frac{2\mu(x)}{\sigma^2(x)} dx}.
\end{align*}
\end{corollary}

This corollary generalizes Theorem \ref{thm:growth_rate} to the time-inhomogeneous case, demonstrating that the long-term growth rate of LP wealth can be expressed in a similar form, with the limiting speed measure $q(y)$ playing a key role.

\subsection{Independent stochastic volatility and drift}
In this section, we generalize the analysis to incorporate stochastic volatility and drift in the log price process, $s_t = \ln S_t$. Specifically, we assume $s_t$ follows the diffusion process:
$$
ds_t = \tilde{\mu}_t dt + \tilde{\sigma}_t dW_t,
$$
where both $\tilde{\mu}_t$ and $\tilde{\sigma}_t$ are stochastic processes but are independent of the driving Brownian motion $W_t$. This allows for more realistic modeling of market dynamics where volatility and expected returns can fluctuate randomly over time. The mispricing process, $Z_t = \ln S_t - \ln P_t$, remains a reflected diffusion within the interval $[-c, c]$, where $c = -\ln\gamma$, and satisfies
\begin{equation*}
dZ_t = ds_t + dL_t - dU_t = \tilde{\mu}_t dt + \tilde\sigma_t dW_t + dL_t - dU_t.
\end{equation*}
We further impose a strong ellipticity condition on the volatility, requiring $\tilde{\sigma}_t \geq \epsilon > 0$ almost surely for all $t$. This ensures that the volatility remains strictly positive.

Our goal is to derive the long-term limit of the time-averaged logarithmic growth rate of LP wealth in this setting. As in Section \ref{sec:time-inhomo}, the eigensystem approach used in Section \ref{sec:time-homo} is not applicable here because the time-dependent drift $\tilde{\mu}_t \neq 0$ prevents the existence of time-independent eigenfunctions.

\subsubsection{Time-Dependent Volatility}
To begin, we consider a simplified scenario where the volatility
$\tilde\sigma_t$ is a deterministic function of $t$ and the drift is zero ($\tilde{\mu}_t = 0$ for all $t$). This allows us to isolate the effect of time-dependent volatility. In this case, the infinitesimal generator becomes time-dependent:
$$
\mathcal{L}_t = \frac{\sigma^2(t)}2 \partial_x^2.
$$
Despite the time-dependence, we can still find eigenvalues and eigenfunctions associated with $\mathcal{L}_t$ that satisfy the eigenvalue problem:
$$
\mathcal{L}_t u = \frac{\sigma^2(t)}2 u_{xx} = -\lambda u
$$
with Neumann boundary conditions $u_x(-c) = u_x(c) = 0$. The solutions are:
\begin{equation} \label{eqn:eigensys-time-dep}
\lambda_n(t) = \frac{\sigma^2(t)}{2} \left(\frac{n\pi}c \right)^2, \qquad e_0(x) = \sqrt{\frac1{2c}}, \qquad
e_n(x) = \sqrt{\frac{1}{c}} \cos\left(\frac{n\pi}{c}x \right) \mbox{ for } n \geq 1.
\end{equation}
Importantly, while the eigenvalues $\lambda_n$ are time-dependent, the eigenfunctions $e_n$ are not.

Using a similar approach as in Section \ref{sec:time-homo}, we can express the conditional expectation
$$
\mathbb{E}[- a (\hat{L}_T - \hat{L}_t) + b (\hat{U}_T - \hat{U}_t)\mid \hat{Z}_t = x]
$$
in terms of the eigensystem \eqref{eqn:eigensys-time-dep}. This leads to the following theorem (presented without proof):

\begin{theorem}
Let $\hat{Z}_t$ be the reflected diffusion defined by \eqref{eqn:zhat-sde} where $\mu = 0$ and $\sigma$ is a deterministic function of $t$.
Then, for given constants $a$ and $b$, the conditional expectation
\begin{equation}
\mathbb{E}\left[- \int_t^T a d\hat{L}_\tau + \int_t^T b d\hat{U}_\tau \, \Big| \, \hat{Z}_t = x \right]
\end{equation}
can be written in terms of the eigensystem \eqref{eqn:eigensys-time-dep} associated with $\cL_t$ with Neumann boundary condition as
\begin{eqnarray}
&& \mathbb{E}\left[- \int_t^T a d\hat{L}_\tau + \int_t^T b d\hat{U}_\tau \, \Big| \, \hat{Z}_t = x \right] = \frac{b-a}{4c} x^2 + \frac{b+a}2 x
+ \sum_{n=0}^\infty v_n(t) e_n(x), \label{eqn:u-eig-expan-time-dep}
\end{eqnarray}
where the time-dependent coefficients $v_n$ are given by
\begin{eqnarray} \label{equation:v_n}
v_n(t) = e^{-\frac12\left(\frac{n\pi}c\right)^2 \int_t^T \sigma^2(s) ds}k_n + \int_t^T h_n(s) e^{-\frac12\left(\frac{n\pi}c\right)^2\int_t^s \sigma^2(\tau) d\tau} ds.
\end{eqnarray}
and
\begin{eqnarray*}
&& h_n(t) = \frac{\sigma^2(t)}{4c} \int_{-c}^c (b-a) e_n(x) dx, \qquad
k_n = -\int_{-c}^c \left(\frac{b-a}{4c} x^2 + \frac{b+a}2 x\right) e_n(x) dx.
\end{eqnarray*}
\end{theorem}

Observe from \eqref{eqn:u-eig-expan-time-dep} that we have
\begin{align*}
\lim_{T\to\infty} \frac1{T-t} v_n(t) &= 0 \quad \forall n \geq 1, \\
\lim_{T\to\infty} \frac1{T-t} v_0(t) &= \lim_{T\to\infty}\frac1{T-t} k_0 + \lim_{T\to\infty}\frac1{T-t} \int_t^T h_0(s) ds \\
&= \frac{b-a}{4c}\, \sqrt{2c} \, \lim_{T\to\infty} \frac1{T-t} \int_t^T \sigma^2(s) ds.
\end{align*}
Consequently,
\begin{eqnarray*}
&& \lim_{T\to\infty} \frac1{T-t} \mathbb{E}\left[- a (\hat{L}_T-\hat{L}_t) + b (\hat{U}_T-\hat{U}_t)\mid \hat{Z}_t \right] \\
&=& \lim_{T\to\infty} \frac1{T-t} \left\{
\sum_{n=0}^\infty v_n(t) e_n(\hat{Z}_t) + \frac{b - a}{4c} \hat{Z}_t^2 + \frac{b+a}2 \hat{Z}_t
\right\} \\
&=& \lim_{T\to\infty} \frac1{T-t} v_0(t) e_0(\hat{Z}_t) \\
&=& \frac{b-a}{4c} \lim_{T\to\infty} \frac1{T-t} \int_t^T \sigma^2(s) ds.
\end{eqnarray*}
This together with Lemma \ref{lem:ez=ezhat} leads to the following theorem.

\begin{theorem}[Log Growth Rate of LP Wealth under Time-Dependent Volatility]
Assuming $\mu_X = \mu_Y = \mu$, the logarithmic growth rate of an LP's wealth in a G3M can be expressed as
$$
\lim_{T \to \infty} \frac{\mathbb{E}[\ln V_T]}{T}
= \mu - \left[\frac{(1-\gamma)w(1-w)}{1-w+\gamma w} + \frac{(1-\gamma)w(1-w)}{\gamma(1-w)+ w}\right] \frac1{4 \ln \gamma} \lim_{T\to\infty} \frac{1}{T}\int_0^T \sigma^2(s) ds.
$$
\end{theorem}

This theorem demonstrates how the time-averaged volatility influences the long-term growth of LP wealth. Furthermore, since the volatility process $\tilde\sigma_t$ is independent of the mispricing process $Z_t$, by conditioning on the $\sigma$-algebra generated by $\tilde\sigma_t$ then applying the tower property for conditional expectation, we obtain the logarithmic growth rate of LP's wealth under driftless, independent stochastic volatility as follows.

\begin{theorem}[Log Growth Rate under Independent Stochastic Volatility]
Assuming $\mu_X = \mu_Y = \mu$, the logarithmic growth rate of an LP's wealth in a G3M can be expressed as
\begin{eqnarray*}
&& \lim_{T \to \infty} \frac{\mathbb{E}[\ln V_T]}{T} \\
&=& \mu - \left[\frac{(1-\gamma)w(1-w)}{1-w+\gamma w} + \frac{(1-\gamma)w(1-w)}{\gamma(1-w)+ w}\right] \frac1{4 \ln \gamma} \lim_{T\to\infty} \frac{1}{T}\int_0^T \Eof{\tilde\sigma_s^2} ds.
\end{eqnarray*}
\end{theorem}

We conclude the section by showing that if the stochastic but independent drift and volatility converge to their corresponding $L^2$ limits, similar asymptotics in expectation as in Corollary \ref{thm:growth_rate_inhomo} can also be obtained.

\begin{theorem}[Log Growth Rate of LP Wealth in G3M Independent Volatility and Drift] \label{thm:growth_rate_stoch}
Assume that the mispricing process $Z_t$ follows the reflected diffusion process in the interval $[-c, c]$ governed by
\begin{equation}
dZ_t = \tilde{\mu}_t dt + \tilde\sigma_t dW_t + dL_t - dU_t,
\end{equation}
where the coefficients $\tilde{\sigma}_t$ and $\tilde{\mu}_t$ are stochastic but independent of the driving Brownian motion $W_t$. Further, assume that the limits of the coefficients $\tilde{\sigma}_t$ and $\tilde{\mu}_t$ as $t$ approaches infinity exist almost surely and in $L^2$. Specifically, there exists an $\epsilon > 0$ such that
\begin{eqnarray}
&& \lim_{t\to\infty} \tilde{\sigma}_t^2 = \sigma^2 \geq \epsilon, \qquad  \lim_{t\to\infty} \tilde{\mu}_t = \mu
\end{eqnarray}
almost surely and in $L^2$, where $\mu$ and $\sigma^2$ are square integrable random variables. The long-term expected logarithmic growth rate of an LP's wealth in a G3M can be expressed as
\begin{align*}
\lim_{T \to \infty} \frac{\mathbb{E}[\ln V_T]}{T}
= w \mu_X + (1-w) \mu_Y + \frac{(1-\gamma)w(1-w)}{1-w+\gamma w} \alpha + \frac{(1 - \gamma)w(1-w)}{\gamma(1-w) + w} \beta,
\end{align*}
where
\begin{eqnarray*}
\begin{array}{ll}
\displaystyle \alpha = \beta = -\frac1{4 \ln \gamma}\Eof{\sigma^2}, & \mbox{ if } \theta = 0 \mbox{ almost surely}; \\
& \\
\displaystyle \alpha = \Eof{\frac{\mu}{\gamma^{-2\theta} - 1}},  \quad \beta = \Eof{\frac{\mu}{1-\gamma^{2\theta}}}, & \mbox{ if } \theta \neq 0,
\end{array}
\end{eqnarray*}
where $\theta = \frac{2 \mu}{\sigma^2}$, should the expectations exist.
\end{theorem}

\begin{proof}
The proof essentially is based on conditioning on the realizations of $\tilde{\mu}_t$ and $\tilde\sigma_t$ followed by applying the tower property since $\tilde{\mu}$ and $\tilde{\sigma}$ are independent of the Brownian motion $W_t$.
\end{proof}

We remark that the long-term expected logarithmic growth rate considered in Theorem \ref{thm:growth_rate_stoch} can also be obtained differently by first calculating the condition expectations as in \eqref{eqn:mu_z-sig_z}, then apply the asymptotic result given in Corollary \ref{thm:growth_rate_inhomo}. This route is applicable even when $\tilde{\mu}_t$ and $\tilde\sigma_t$ are not independent of $W_t$; it is, however, subject to the determination of the conditional expectations in \eqref{eqn:mu_z-sig_z}, which in general do not admit easy-to-access analytical expressions. The expression obtained in Theorem \ref{thm:growth_rate_stoch} is more tractable in that it is subject to the determination of the limiting distributions for $\mu$ and $\sigma^2$ as well as the corresponding expectations. However, it applies only if $\tilde{\mu}_t$ and $\tilde{\sigma}_t$ are independent of $W_t$.
\section*{Conclusion}
This paper studied the long-term growth of liquidity provider (LP) wealth in
Geometric Mean Market Makers (G3Ms), explicitly account for the effects of
continuous-time arbitrage and transaction fees. Building on the framework of
\cite{milionis2022quantifying, milionis2022automated, milionis2023automated} and extending the constant product analysis of \cite{tassy2020growth} to a broader class of G3Ms, we developed a tractable stochastic model driven by the mispricing process between the G3M pool and an external reference market.

A central finding is that trading fees play a dual role in LP wealth dynamics: they generate income through the no-arbitrage band they impose on the pool price, while simultaneously mediating the adverse selection risk posed by arbitrageurs. Our explicit formulas for the long-term logarithmic growth rate of LP wealth capture both effects, and reveal that --- under appropriate fee selection --- G3Ms can outperform both buy-and-hold and constant-rebalanced portfolio strategies. The result connects G3M dynamics to the excess growth rate in the Stochastic Portfolio Theory \cite{fernholz2002stochastic}, and supports the view of G3Ms as viable on-chain index fund infrastructure: short-term arbitrage losses are transformed, through the fee mechanism, into long-run volatility harvesting.

The present analysis is carried out in the two-asset setting, which allows for a complete and tractable treatment of the mispricing dynamics and their ergodic properties. The structural form of these dynamics, however, extends naturally to the multi-asset case. For a pool with $n$ assets, the log-mispricing of asset $i$ relative to asset $j$ satisfies
$$
dZ_t^{ij} = d\ln S_t^{ij} + \sum_{k \neq i} \alpha^{ij}_{ik} dL_t^{ik} - \sum_{k \neq j} \beta^{ij}_{kj} dU_t^{kj},
$$
where $\alpha^{ij}_{ik}$ and $\beta^{ij}_{kj}$ are constants determined by the fee parameter $\gamma$, the pool weights $w$, and the normalization of order flows, and $L_t^{ij}$, $U_t^{ij}$ are continuous non-decreasing processes that increase only when $Z_t^{ij} = \ln\gamma$ and $Z_t^{ij} = -\ln\gamma$, respectively. The principal challenge in extending our results to this setting lies not in the dynamics' structural form, which carries over directly, but in characterizing the ergodic behavior of the boundary local times $L_t^{ij}$ and $U_t^{ij}$. These are driven by a higher-dimensional reflected diffusion for which the existing literature does not yet provide a complete theory. We intend to pursue this direction in subsequent work.

Several further directions emerge naturally from this work. These include incorporating noise-trader order flow to capture the contribution of uninformed trading to LP profitability, investigating the interplay between multiple AMM liquidity pools and a fragmented market microstructure, and extending the framework to G3Ms with dynamic weights in the spirit of \cite{evans2021liquidity, fernholz2002stochastic, karatzas2009stochastic}. Developing the ergodic theory for multi-asset reflected diffusions and deriving the corresponding LP wealth growth rates for pools with $n > 2$ assets constitutes the most immediate and natural direction for future research, which we intend to pursue in subsequent work.

\section*{Acknowledgement}
The authors express their sincere gratitude to Shuenn-Jyi Sheu for his invaluable insights and guidance throughout this research. We also appreciate the fruitful discussions and support from our colleagues, which significantly contributed to the development of this work. We are also grateful to the anonymous referee for their careful reading of the manuscript and their constructive comments and suggestions, which helped improve the quality of this paper.

S.-N. T. gratefully acknowledges the financial support from the National Science and Technology Council of Taiwan under grant 111-2115-M-007-014-MY3. Furthermore, S.-N. T. extends heartfelt thanks to Ju-Yi Yen for her unwavering encouragement and support, which were instrumental in making this collaborative effort possible.

C. Y. Lee was supported in part by a research startup fund of the Chinese University of Hong Kong, Shenzhen, and the Shenzhen Peacock fund 2025TC0013.
\appendix

\section{Uniqueness for the Fokker-Planck equation}\label{append:uniqueness}
Here we provide details for the uniqueness of solutions $p$ and $\hat{p}$ in the proof of Lemma \ref{lem:z=zhat}. Let $q(t,z) = p(t,z) - \hat{p}(t,z)$. Then $q$ solves the PDE
\begin{align}\label{FPE0}
    \begin{cases}
        q_t = -(\mu q)_z + \frac12 (\sigma^2 q)_{zz}, & t > 0, z \in (-c,c),\\
        \mu q - \frac12 (\sigma^2 q)_{zz} = 0 & \text{at $z=\pm c$,}\\
        q(0,z) = 0 & z \in [-c,c].
    \end{cases}
\end{align}
We use the energy method to show that $q=0$. Let
\begin{align*}
    E(t) = \frac12 \int_{-c}^c q^2(t,z) dz.
\end{align*}
Using the PDE in \eqref{FPE0}, integrating by parts, and taking into account the boundary condition in \eqref{FPE0}, we have for any $t \in [0,T]$ that
\begin{align*}
    E'(t) &= \int_{-c}^c q q_t dz = -\int_{-c}^c q (\mu q)_z dz + \frac12 \int_{-c}^c q (\sigma^2 q)_{zz} dz\\
    & = \int_{-c}^c \mu q q_z dz - \frac12 \int_{-c}^c q_z (\sigma^2 q)_z dz\\
    & = \int_{-c}^c (\mu-\sigma \sigma_z) q q_z dz - \frac12 \int_{-c}^c \sigma^2 (q_z)^2 dz.
\end{align*}
By assumption, $\mu$, $\sigma$ and $\sigma_z$ are bounded.
Moreover, by the strong ellipticity \eqref{strong_ellipticity} of $\sigma$ and Young's inequality, we deduce that for some $C>0$,
\begin{align*}
    E'(t) &\le C\int_{-c}^c |q| |q_z| dz - \frac{\varepsilon}{2} \int_{-c}^c (q_z)^2\\
    & \le \frac{C^2}{2\varepsilon} \int_{-c}^c q^2 dz + \frac{\varepsilon}{2} \int_{-c}^c (q_z)^2 - \frac{\varepsilon}{2} \int_{-c}^c (q_z)^2 dz\\
    &\le \frac{C^2}{2\varepsilon} \int_{-c}^c q^2 dz = \frac{C^2}{\varepsilon}E(t).
\end{align*}
Hence, Gronwall's inequality implies that $E(t) \le e^{(C^2/\varepsilon)t}E(0) = 0$ for any $t \in [0,T]$. This shows that $p(t,\cdot) = \hat{p}(t,\cdot)$.

\section{Sturm-Liouville theory} \label{append:S-L}
This section provides a brief overview of Sturm-Liouville theory, which plays a key role in analyzing the mispricing process and deriving the growth rate of LP wealth in our framework.

\subsection{Eigensystem} \label{sec:eigensystem}
Consider a second-order differential operator $\mathcal{L}$ in the Sturm-Liouville form
\begin{equation} \label{eqn:L-sl}
\mathcal{L} u = \frac1{\omega(x)} \frac{d}{dx}\left( p(x) \frac{du}{dx}\right),
\end{equation}
where $\omega(x)$ and $p(x)$ are smooth functions. Let $(\lambda_n, e_n(x))$ represent the eigensystem of $\mathcal{L}$, satisfying
$$
\mathcal{L} e_n(x) = -\lambda_n e_n(x) \quad \text{for } n \geq 0,
$$
with Neumann boundary conditions $e_n'(-c) = e_n'(c) = 0$ on the interval $[-c,c]$. These boundary conditions correspond to reflecting boundaries for the associated diffusion process. Since a constant function always satisfies the equation with $\lambda_0 = 0$, the first normalized eigenfunction is $e_0 \equiv \frac{1}{K}$, where $K = (\int_{-c}^c \omega(x) dx)^{\frac12}$.

The eigensystem possesses the following important properties:
\begin{itemize}
    \item $\lambda_n > 0$ for all $n \in \mathbb{N}_{>0}$ and each eigenvalue is of multiplicity one.
    \item The normalized eigenfunctions $e_n$ form an orthonormal basis for the space $L^2[-c, c]$ with respect to the weight function $\omega(x)$. This means
    $$
    \int_{-c}^c e_n(x)e_m(x) \omega(x)dx = \delta_{nm}
    $$
    where $\delta_{nm}$ is the Kronecker delta. Consequently, any function $f \in L^2[-c, c]$ can be expressed as
    $$
    f(x) = \sum_{n=0}^\infty \xi_n e_n(x),
    $$
    where the coefficients $\xi_n$ are given by
    $$
    \xi_n = \int_{-c}^c f(x)e_n(x)\omega(x)dx.
    $$
    This expansion converges in the $L^2$ sense with respect to the weight function $\omega(x)$.
\end{itemize}

\subsection{Solution to inhomogeneous PDE with general terminal condition}
For the inhomogeneous parabolic PDE
\begin{equation}
u_t + \mathcal{L} u + f(x) = 0 \label{eqn:inhom-pde}
\end{equation}
with terminal condition $u(T,x) = g(x)$ and Neumann boundary conditions $u_x(t,-c) = u_x(t, c) = 0$ for $t < T$, we show how to formulate its solution using the eigensystem given in Section \ref{sec:eigensystem}.
Let the eigenfunction expansions for functions $f$ and $g$ be represented as
$$
f(x) = \sum_{n=0}^\infty \xi_n e_n(x), \qquad g(x) = \sum_{n=0}^\infty \eta_n e_n(x)
$$
where the coefficients $\xi_n$ and $\eta_n$ are defined as:
$$
\xi_n = \int_{-c}^c f(x)e_n(x)\omega(x)dx, \qquad \eta_n = \int_{-c}^c g(x)e_n(x)\omega(x)dx.
$$
In particular, the term $\xi_0$, expressed as
\begin{equation} \label{eqn:xi0}
\xi_0 = \frac{\int_{-c}^c f(x)\omega(x) dx}{\int_{-c}^c \omega(x)dx},
\end{equation}
is the weighted average of $f$ over the interval $[-c, c]$, weighted by $\omega$.

The solution to the terminal-boundary value problem \eqref{eqn:inhom-pde}
can be expressed in terms of eigenvalues and eigenfunctions for $\cL$ as
\begin{equation}
u(t, x) = \xi_0 (T - t) + \eta_0 + \sum_{n=1}^\infty \left\{\frac{\xi_n}{\lambda_n} \left[1 - e^{-\lambda_n(T - t)} \right] + \eta_n e^{-\lambda_n(T - t)} \right\} e_n(x).
\end{equation}
Consequently, the following long-term time-averaged limit of $u$ exists
$$
\lim_{T\to\infty} \frac{u(t,x)}{T-t} = \xi_0,
$$
where recall that $\xi_0$, given in \eqref{eqn:xi0}, is the zeroth Fourier coefficient of the inhomogeneous term $f$. We note that this long-term time-averaged limit depends only on the zeroth coefficient of the inhomogeneous term; no other higher-order coefficients are involved. Furthermore, we have the following long-term limit of $u$ as $T \to \infty$
$$
\lim_{T\to\infty} \left\{ u(t,x) - \xi_0(T - t) \right\} = \eta_0 + \sum_{n=1}^\infty \frac{\xi_n}{\lambda_n} e_n(x).
$$

\subsection{Transition density in terms of eigensystem}
The following proposition shows that the transition density of a reflected diffusion \eqref{eqn:zhat-sde} within a bounded interval can be expressed in terms of the eigensystem of its infinitesimal generator with Neumann boundary conditions.

\begin{proposition} \label{Prop:transition_reflected_diffusion}
The transition density $p$ of a reflected diffusion in the interval $[-c, c]$ with infinitesimal generator $\cL$ given in \eqref{eqn:L-sl} can be expressed in terms of the eigensystem for $\cL$ as
$$
p(T,y|t,x) = \sum_{n=0}^\infty e^{-\lambda_n(T - t)} e_n(x) e_n(y) \omega(y).
$$
\end{proposition}

This leads to the following characterization of the steady-state distribution:

\begin{theorem}[Steady-State Distribution] \label{thm:steady_state_RD}
The reflected diffusion within the interval $[-c,c]$ with infinitesimal generator \eqref{eqn:L-sl} has a steady-state distribution $\pi$ given by
$$
\pi(dx) = \frac{\omega(x)}{\int_{-c}^c \omega(\xi) d\xi} dx, \quad x \in [-c,c].
$$
\end{theorem}

\begin{proof}
By Proposition \ref{Prop:transition_reflected_diffusion}, as $T \to \infty$, the steady-state distribution is given by
\begin{align*}
\lim_{T\to\infty} p(T,y|t, x)
&= \lim_{T\to\infty}\sum_{n=0}^\infty e^{-\lambda_n(T - t)} e_n(x) e_n(y) \omega(y) \\
&= e_0(x) e_0(y) \omega(y)
= \frac{\omega(y)}{\int_{-c}^c \omega(x) dx}
\end{align*}
since the zeroth eigenfunction $e_0(x)$ is a constant $e_0(x) = \left(\int_{-c}^c \omega(\xi)d\xi\right)^{-\frac12}$.
\end{proof}

\section{Time-inhomogeneous reflected diffusion} \label{append:time-inhomo}
In this appendix, we provide the proof of the long-term time averaged expectation of a time-inhomogeneous reflected diffusion as stated in Theorem \ref{thm:time-inhomo-lim}. For fixed $t, x$, we shall sometimes suppress the reference to $t$, $x$ in the transition density $p$ and simply denote $p(s,y|t,x)$ by $p(s,y)$ for simplicity. For any function $\varphi$ defined in $[-c,c]$, $\|\varphi\|_2$ denotes the $L^2$ norm of $\varphi$ in $[-c,c]$. We start with stating an estimate of the $L^2$-norm between the transition density $p$ and the stationary density $q$ in the following lemma, whose proof is omitted (for interested readers, we refer it to, for instance, \cite{kahane1983asymptotic}, see (3.21) on P.276), is classical and crucial to the proof that follows.

\begin{lemma} \label{lma:L2-ergodicity}
Assume that the infinitesimal generator operator $\cL_t := \frac{\sigma^2(t,x)}2 \p_x^2 + \mu(t,x)\p_x$ is strongly elliptic, i.e., there exists an $\epsilon > 0$ such that $\sigma(t,x) \geq \epsilon$ for all $t$, $x$, and that the coefficients $\sigma$ and $\mu$ are smooth and bounded,
the following estimates hold. For any $T > t$, we have
\begin{equation} \label{eqn:L2-ergodicity}
\|p(T,\cdot|t,x) - q\|_2 \leq \frac C{\sqrt{T - t}}
\end{equation}
for some constant $C$ depending only on the interval $[-c,c]$. As a result, we note that the $L^2$ norm of $p(s, y)$ is bounded above by
\begin{equation} \label{eqn:p-ub}
\|p(s,\cdot)\|_2 \leq \|p(s,\cdot) - q\|_2 + \|q\|_2 \leq \frac C{\sqrt{s - t}}+ \|q\|_2
\end{equation}
for $s > t$.
\end{lemma}

With Lemma \ref{lma:L2-ergodicity} in hand, we provide the proof of Theorem \ref{thm:time-inhomo-lim} as follows. Note that we have
\begin{align*}
u(t, x) &= \mathbb{E}\left[g(X_T) + \int_t^T f(s, X_s)ds\, \Big| \, \hat{Z}_t = x \right] \\
&= \int_{-c}^c g(y)p(T,y|t,x) dy + \int_t^T\int_{-c}^c f(s,y)p(s,y|t,x)dyds
\end{align*}
since $p$ is the transition density. Consider
\begin{equation} \label{eqn:u-inhomo}
\frac{u(t, x)}{T-t} = \frac1{T-t}\int_{-c}^c g(y)p(T,y|t,x) dy + \frac1{T - t}\int_t^T\int_{-c}^c f(s,y)p(s,y|t,x)dyds,
\end{equation}
We separately determine the limits of the two terms on the right-hand side of \eqref{eqn:u-inhomo}.

For the first term in \eqref{eqn:u-inhomo}, by applying the Cauchy-Schwarz inequality we obtain that, for $T \geq t$,
\begin{align*}
& \left|\int_{-c}^c g(y) p(T,y|t, x) dy \right| \leq  \|g\|_2 \|p(T,\cdot|t, x)\|_2  \leq  \|g\|_2 \left\{\frac C{\sqrt{T-t}} + \|q\|_2 \right\}
\end{align*}
where in the second inequality, we applied the upper bound for $p(T,\cdot)$ given in \eqref{eqn:p-ub}. It follows that
\begin{equation} \label{eqn:1st-term-lim}
\lim_{T\to\infty}\frac1{T - t}\left|\int_{-c}^c g(y) p(T,y|t, x) dy\right|
\leq \lim_{T\to\infty}\frac{\|g\|_2}{T-t} \left\{\frac C{\sqrt{T-t}} + \|q\|_2 \right\} = 0.
\end{equation}

For the second term in \eqref{eqn:u-inhomo}, we claim that, as $\lim_{t\to\infty}f(t, x) = \bar f(x)$ in $L^2$, we have
\begin{equation} \label{eqn:2nd-term-lim}
\lim_{T\to\infty}\frac1{T-t}\int_t^T \int_{-c}^c f(s,y)p(s,y|t,x) dy ds = \int_{-c}^c \bar f(y)q(y) dy
\end{equation}
for every $t$ and $x$. Note that by applying the Cauchy-Schwarz inequality, we have
\begin{align}
& \left|\int_{-c}^c \left\{f(s,y)p(s,y) - \bar f(y)q(y)\right\} dy \right|  \nonumber \\
\leq& \int_{-c}^c \left|f(s,y) - \bar f(y)\right|p(s,y) dy + \int_{-c}^c |\bar f(y)||p(s,y)-q(y)| dy \nonumber \\
\leq& \|f(s,\cdot) - \bar f\|_2 \|p(s,\cdot)\|_2 + \|\bar f\|_2 \|p(s,\cdot) - q \|_2.\label{eqn:2nd-term}
\end{align}
We shall deal with the two pieces in \eqref{eqn:2nd-term} separately. For the first piece, since $f(s,y) \to \bar f(y)$ as $s\to\infty$ in $L^2$, for any $\epsilon > 0$, there exists a $t_1 \geq t$ such that
$$
\|f(s,\cdot) - \bar f\|_2 < \epsilon
$$
for all $s > t_1$. Hence, for given $T > t_1$ we have
\begin{eqnarray*}
&& \int_t^T \|f(s,\cdot)-\bar f\|_2 \|p(s,\cdot)\|_2 ds \\
&=& \int_t^{t_1} \|f(s,\cdot)-\bar f\|_2 \|p(s,\cdot)\|_2 ds + \int_{t_1}^T \|f(s,\cdot)-\bar f\|_2 \|p(s,\cdot)\|_2 ds \\
&\leq& M \int_t^{t_1} \left\{ \frac C{\sqrt{s-t}} + \|q\|_2\right\} ds + \epsilon \int_{t_1}^T \left\{ \frac C{\sqrt{s-t}} + \|q\|_2\right\} ds \\
&=& M \left\{ 2C \sqrt{t_1-t} + \|q\|_2 (t_1 - t) \right\} + \epsilon \left\{ 2 C \sqrt{T-t_1} + \|q\|_2 (T - t_1)\right\} ,
\end{eqnarray*}
where in the inequality we applied the upper bound for $p$ given in \eqref{eqn:p-ub} and the constant $M > 0$ is defined as
$$
\infty > M := \max_{t\leq s\leq t_1}\|f(s,\cdot)\|_2 + \|\bar f\|_2 \geq \|f(s,\cdot)\|_2 + \|\bar f\|_2 \geq \|f(s,\cdot)-\bar f\|_2.
$$ It follows that
\begin{eqnarray*}
&& \lim_{T\to\infty}\frac1{T - t}\int_t^T \|f(s,\cdot)-\bar f\|_2 \|p(s,\cdot)\|_2 ds \\
&\leq& \lim_{T\to\infty} \frac M{T-t} \left\{ 2C \sqrt{t_1-t} + \|q\|_2 (t_1 - t) \right\} + \lim_{T\to\infty} \frac\epsilon{T-t} \left\{ 2 C \sqrt{T-t_1} + \|q\|_2 (T - t_1)\right\} \\
&=& \epsilon \|q\|_2.
\end{eqnarray*}
Since $\epsilon > 0$ is arbitrary, we obtain the limit of time-average of the first piece in \eqref{eqn:2nd-term} as $T$ approaches infinity as
\begin{align*}
\lim_{T\to\infty}\frac1{T - t}\int_t^T \|f(s,\cdot)-\bar f\|_2 \|p(s,\cdot)\|_2 ds = 0.
\end{align*}
Next, for the second piece in \eqref{eqn:2nd-term}, note that from \eqref{eqn:L2-ergodicity} we have
\begin{align*}
& \int_t^T \|p(s,\cdot) - q \|_2 ds \leq \int_t^T \frac C{\sqrt{s-t}} ds = 2 C \sqrt{T - t}.
\end{align*}
It follows immediately that
\begin{align*}
&\lim_{T\to\infty} \frac1{T-t} \int_t^T \|\bar f\|_2\|p(s,\cdot) - q \|_2 ds
\leq 2 C \lim_{T\to\infty} \frac{\|\bar f\|_2}{T-t}  \sqrt{T - t}
= 0.
\end{align*}
Finally, by combing \eqref{eqn:1st-term-lim} and \eqref{eqn:2nd-term-lim} we conclude that
\begin{align*}
\lim_{T\to\infty}\frac1{T-t} u(t,x) = \int_{-c}^c \bar f(y)q(y)dy.
\end{align*}

%-------------------------------------------
%  References
%-------------------------------------------
\bibliographystyle{alpha}
\bibliography{Reference}
\end{document}